\documentclass[pra,showpacs,amsmath,amssymb,aps,superscriptaddress,twocolumn]{revtex4-1}
\usepackage{graphicx} 
\usepackage{amsthm}
\usepackage{braket}
\usepackage{bm}
\usepackage{here}
\usepackage{hyperref}
\usepackage{ascmac}
\usepackage{siunitx}
\usepackage{stmaryrd}
\usepackage{mathtools} 
\usepackage{color}
\DeclareMathOperator\arctanh{arctanh}
\newtheorem{theorem}{Theorem}
\newtheorem{lemma}{Lemma}
\newtheorem{prop}{Proposition}
\newtheorem{cor}{Corollary}
\newtheorem{definition}{Definition}

\allowdisplaybreaks[2]

\begin{document} 

\title{Equivalence of approximate Gottesman-Kitaev-Preskill codes}
\author{Takaya Matsuura}
 \email{matsuura@qi.t.u-tokyo.ac.jp}
 \affiliation{Department of Applied Physics, Graduate School of Engineering, The University of Tokyo, 7-3-1 Hongo, Bunkyo-ku, Tokyo 113-8656, Japan} 
 \author{Hayata Yamasaki}
 \affiliation{Photon Science Center, Graduate School of Engineering, The University of Tokyo, 7-3-1 Hongo, Bunkyo-ku, Tokyo 113-8656, Japan} 
\author{Masato Koashi}
 \affiliation{Department of Applied Physics, Graduate School of Engineering, The University of Tokyo, 7-3-1 Hongo, Bunkyo-ku, Tokyo 113-8656, Japan} 
 \affiliation{Photon Science Center, Graduate School of Engineering, The University of Tokyo, 7-3-1 Hongo, Bunkyo-ku, Tokyo 113-8656, Japan} 
 \date{\today}

\begin{abstract} 
    The Gottesman-Kitaev-Preskill (GKP) quantum error-correcting code attracts much attention in continuous variable (CV) quantum computation and CV quantum communication due to the simplicity of error-correcting routines and the high tolerance against Gaussian errors.  Since the GKP code state should be regarded as a limit of physically meaningful approximate ones, various approximations have been developed until today, but explicit relations among them are still unclear.  In this paper, we rigorously prove the equivalence of these approximate GKP codes with an explicit correspondence of the parameters.  We also propose a standard form of the approximate code states in the position representation, which enables us to derive closed-form expressions for the Wigner function, inner products, and the average photon number in terms of the theta functions.  Our results serve as fundamental tools for further analyses of fault-tolerant quantum computation and channel coding using approximate GKP codes.  
\end{abstract}

\maketitle

\section{Introduction}
Continuous variable (CV) systems \cite{Braunstein2005,Gerd2007,Weedbrook2012,Serafini2017} have attracted a growing interest in the field of quantum information science as promising candidates for implementing quantum information processing.  For reliable implementations of information processing tasks, one needs to construct an error-correcting routine to fight against the inevitable noise in the real world.  Intensive research has thus been made on CV error-correcting codes \cite{Lloyd1998,Braunstein1998,Gottesman2001,Menicucci2014,Ketterer2016,Cochrane1999,Niset2008,Leghtas2013,Lacerda2016,Lacerda2017,Chuang1997,Knill2001,Ralph2005,Wasilewski2007,Bergmann2016,Michael2016,Niu2018,Albert2018}.  Among them, the Gottesman-Kitaev-Preskill (GKP) code \cite{Gottesman2001} gathers much attention in terms of both fault-tolerant CV quantum computation \cite{Menicucci2014,Douce2017,Fukui2017,Fukui2018,Vuillot2019,Walshe2019,Wang2019,Noh2019,Fukui2019,Tzitrin2019,Hanggli2020} and CV quantum communication \cite{Harrington2001,Albert2018,Noh2018} as it needs only Gaussian operations to implement Clifford gates (or even the universal gate set using protocols with a single GKP code state \cite{Baragiola2019,Yamasaki2019}), and it is highly robust against random displacement errors and loss errors \cite{Caruso2006}.

The ideal GKP code state is non-normalizable, while physically meaningful states in quantum mechanics are normalizable.  Therefore, we have to regard an ideal GKP code state as a limit of an approximate code state.
Various approximations of the GKP code states, which are considered to be roughly equivalent, appeared in the past literature \cite{Gottesman2001,Pirandola2004,Glancy2006,Vasconcelos2010,Menicucci2014,Albert2018,Noh2018,Weigand2018,Tzitrin2019}, each of which uses a convenient form of approximation in its respective context.  However, exact relations between these approximations are unclear, and thus we lack a way to compare these results directly.  

Our aim here is to find rigorous relations among the different approximations of the GKP code states, and bridge the gap of the results in the past literature.  We derive an explicit correspondence among conventionally used approximate GKP code states.  The explicit formula shows that one of the conventionally used approximations that has been considered to be symmetric in position and momentum coordinates in phase space is in fact asymmetric.  We also derive closed-form expressions of the Wigner function, normalization constant, and the average photon number of these approximate code states.  These results show that around the degree of approximation for the code states that have been successfully generated in recent experiments \cite{Fluhmann2019,Campagne2019}, conventional estimates of the average photon number of the code state have non-negligible error.  In contrast, our results are accurate in all the degrees of approximation. 

This paper is organized as follows.  In Sec.~\ref{sec:notation}, we define the notation used throughout this paper.  In Sec.~\ref{sec:formulation}, we review the formulation of the GKP code, and introduce its three approximations which have been conventionally used.  In Sec.~\ref{sec:main_result}, which contains the main results of our paper, we explicitly give the position and momentum representations of these approximate code states.  They allow us to derive the exact relations among these approximate code states as shown in Theorem \ref{theorem:main}.  Using the equivalence, we introduce a standard form of the approximate GKP code state.  In Sec.~\ref{sec:applications}, we derive Wigner function, inner products, and the average photon number of the approximate code states using the standard form.  Finally in Sec.~\ref{sec:conclusion}, we give concluding remarks.

\section{Notation} \label{sec:notation}
Canonical operators are denoted by $\hat{q}$ and $\hat{p}$, which satisfy the commutation relation $[\hat{q},\hat{p}]=i$, where we set $\hbar=1$.  Annihilation and creation operators are denoted by $\hat{a}$ and $\hat{a}^{\dagger} $, respectively, which are associated with $\hat{q}$ and $\hat{p}$ as $\hat{q}=(\hat{a} + \hat{a}^{\dagger})/\sqrt{2}$ and $\hat{p}=(\hat{a} - \hat{a}^{\dagger})/(\sqrt{2}i)$.  This leads to the commutation relation $[\hat{a},\hat{a}^{\dagger}]=1$.  The Weyl-Heisenberg displacement operators are represented by $\hat{X}(r)\coloneqq \exp(-ir\hat{p})$ and $\hat{Z}(r)\coloneqq \exp(ir\hat{q})$, which displace a state by $+r$ in position and momentum coordinates in phase space, respectively.  General Weyl-Heisenberg displacement operators are represented by $\hat{V}(\bm{r})\coloneqq \exp(-ir_p r_q /2)\hat{Z}(r_p)\hat{X}(r_q)$, where $\bm{r}=(r_p,r_q)$.  The relation between $\hat{V}(\bm{r})$ and the conventional definition of the displacement operator $\hat{D}({\alpha})\coloneqq \exp(\alpha\hat{a}^{\dagger} - \alpha^* \hat{a})$ \cite{Leonhardt1965} is $\hat{V}(\bm{r}) = \hat{D}((r_q+ir_p)/\sqrt{2})$.  The squeezing operator $\hat{S}(\xi)$ $(\xi\in\mathbb{R})$ is defined as 
$
    \hat{S}(\xi)\coloneqq \exp\left(i \xi \left(\hat{q}\hat{p} + \hat{p}\hat{q}\right)/2\right),
$ 
which satisfies $\hat{S}^{\dagger}(\xi)\hat{q}\hat{S}(\xi)= e^{-\xi}\hat{q}$ and $\hat{S}^{\dagger}(\xi)\hat{p}\hat{S}(\xi)= e^{\xi}\hat{p}$ \cite{Leonhardt1965}.  The number operator $\hat{n}$ is defined as $\hat{n}\coloneqq \hat{a}^{\dagger}\hat{a}$, and the Fourier operator $\hat{F}$ is defined as $\hat{F}\coloneqq \exp\left(\pi i \hat{n} / 2\right)$.  Let $\hat{I}$ denote the identity operator.

Throughout the paper, $\ket{\cdot}$ denotes the logical states of (approximate) GKP codes.  Other representations are specified by subscripts of ket vectors.  For example, $\ket{n}_f$ denotes the Fock state,
$\ket{q}_{\hat{q}}$ denotes the (generalized) eigenstate of the position operator $\hat{q}$, and $\ket{p}_{\hat{p}}$ is that of the momentum operator $\hat{p}$.  The latter two satisfy ${}_{\hat{q}}\braket{q|q'}_{\hat{q}} = \delta(q-q')$, ${}_{\hat{p}}\braket{p|p'}_{\hat{p}} = \delta(p-p')$, ${}_{\hat{q}}\braket{q|p}_{\hat{p}} = \frac{1}{\sqrt{2\pi}}e^{iqp}$, and $\hat{F}\ket{x}_{\hat{q}}=\ket{x}_{\hat{p}}$, where $\delta(\cdot)$ denotes the Dirac delta function.

We also line up functions that are used throughout the paper.  For $z\in\mathbb{C}$ and $\tau\in\mathbb{C}$ satisfying $\mathrm{Im}(\tau)>0$, let $\vartheta(z,\tau)\coloneqq \sum_{s\in\mathbb{Z}}\exp(\pi i \tau s^2 + 2\pi i z s)$ be the theta function (we follow the notation in Ref.~\cite{Mumford2007}), and  
\begin{align}
&\vartheta \! \left[\begin{subarray}{c} a \\ \ \\ b \end{subarray} \right]\! (z,\tau) \nonumber \\
&\coloneqq \sum_{s\in\mathbb{Z}}\exp[\pi i \tau (s+a)^2 + 2\pi i (z+b) (s+a)]  \\
&=\exp[\pi i \tau a^2 + 2\pi i a(z+b)]\, \vartheta(z+\tau a+b,\tau) 
\label{eq:theta_rational}
\end{align}
be the theta function with rational characteristics $(a,b)$ \cite{Mumford2007}.
\begin{figure}[tbp]
    \centering
    \includegraphics[width=0.95\linewidth]{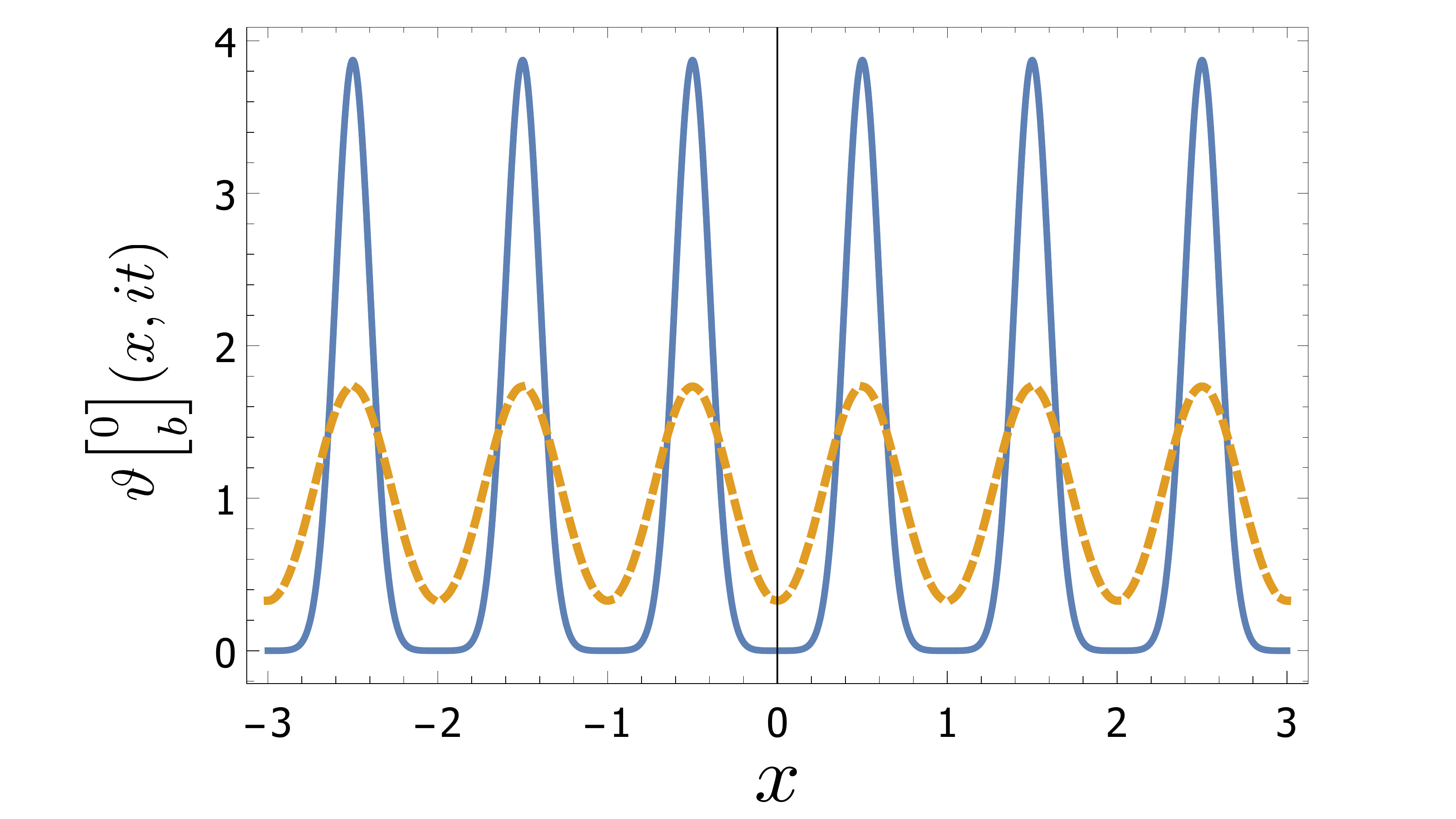}
    \caption{The theta function in the form of Eq.~\eqref{eq:theta_this_paper} with respect to $x$ when $b=1/2$ and $t=1/15$ (blue solid line), and when $b=1/2$ and $t=1/3$ (yellow dashed line).  The theta function in this form is a sequence of the same Gaussian functions with respect to $x$ which has peaks at $b,b\pm1,b\pm2,\ldots$, and the width of each Gaussian is determined by $t$ as shown in the figure.  Note that Eq.~\eqref{eq:theta_this_paper} approaches the Dirac comb as $t\rightarrow 0$.}
    \label{fig:theta_func}
\end{figure}
The theta functions which we mainly use are in the form 
\begin{equation}
\vartheta \! \left[\begin{subarray}{c} 0 \\ \ \\ b \end{subarray} \right]\! (x,it),
    \label{eq:theta_this_paper}
\end{equation}
where $x,t\in\mathbb{R}$, and $b\in\mathbb{Q}$.  The theta function in this form is a sequence of the same Gaussian functions with respect to $x$ which has peaks at $b,b\pm1,b\pm2,\ldots$, and the width of each Gaussian is determined by $t$ as shown in Fig.~\ref{fig:theta_func}.   
Note that Eq.~\eqref{eq:theta_this_paper} approaches the Dirac comb as $t\rightarrow 0$. 
Let $G_{\sigma^2}(x)$ be a probability density function of the normal distribution with variance $\sigma^2$, which is defined as
\begin{equation}
    G_{\sigma^2}(x)\coloneqq \frac{1}{\sqrt{2\pi\sigma^2}}\exp\left(-\frac{x^2}{2\sigma^2}\right).
    \label{eq:gaussian}
\end{equation}
For an operator $\hat{A}$ acting on a Hilbert space, the Wigner function $W_{\hat{A}}(q,p)$ of $\hat{A}$ is given by
\begin{equation}
    W_{\hat{A}}(q,p)= \frac{1}{\pi}\int_{-\infty}^{\infty} dx\, e^{2ipx} {}_{\hat{q}}\braket{q-x|\hat{A}|q+x}_{\hat{q}}.
\end{equation}
Finally, let $f*g(x)\coloneqq \int \! dy\, f(y)g(x-y)$ denote the convolution of two functions $f(x)$ and $g(x)$.

\section{The Gottesman-Kitaev-Preskill code} \label{sec:formulation}
The Gottesman-Kitaev-Preskill (GKP) code \cite{Gottesman2001} is an error-correcting code which encodes $d$-dimensional logical Hilbert space into an oscillator mode.  It has a lattice-like periodic structure when represented in phase space; the Wigner function of the code states $\ket{j}$ and $\ket{j+1}$ have the same period but $\ket{j+1}$ is shifted from $\ket{j}$ by $\frac{1}{d}$ of the period in position.
In the present paper, we treat the square lattice GKP code; it is possible to generalize our results to the hexagonal lattice GKP code.
The ideal (square lattice) GKP code states are defined as \cite{Gottesman2001}
\begin{equation}
    \ket{j^{(\mathrm{ideal})}}\coloneqq \sqrt{\alpha d}\sum_{s\in \mathbb{Z}}\ket{\alpha(ds + j)}_{\hat{q}},
    \label{eq:definition_ideal}
\end{equation}
where $d$ denotes the dimension of the logical Hilbert space, $j\in \{0,\ldots,d-1\}$, and the pre-factor $\sqrt{\alpha d}$ is for later convenience.  In position representation, it has a comb-like shape consisting of the Dirac delta functions (i.e., a Dirac comb) at intervals $\alpha d$, and $\ket{j+1^{(\mathrm{ideal})}}$ is shifted from $\ket{j^{(\mathrm{ideal})}}$ by $\alpha$.  These states form a basis of the $d$-dimensional logical Hilbert space in an oscillator system, and therefore, we call them ideal logical basis states.  
In the momentum representation, the logical basis states are given by
\begin{align}
    \ket{j^{(\mathrm{ideal})}}&= \int \! dy\, \sqrt{\alpha d}\sum_{s\in \mathbb{Z}}\ket{y}_{\hat{p}}\braket{y|\alpha(ds + j)}_{\hat{q}}  \\
    &= \sqrt{\frac{\alpha d}{2\pi}}\int \!dy\, \sum_{s\in \mathbb{Z}} e^{-i\alpha(ds + j) p}\ket{y}_{\hat{p}} \\
    &= \sqrt{2\pi \alpha d} \int \!dy\, \sum_{t\in \mathbb{Z}} \delta(\alpha d y - 2\pi t)e^{-ij\alpha p}\ket{y}_{\hat{p}} \\ 
    &= \sqrt{\frac{2\pi}{\alpha d}}\sum_{t\in \mathbb{Z}} e^{-i\frac{2\pi jt}{d}}\ket{{2\pi t}/(\alpha d)}_{\hat{p}}, \label{eq:momentum_ideal}
\end{align}
where we used the Poisson summation formula $\sum_{s\in \mathbb{Z}}e^{-isx} = 2\pi\sum_{t\in\mathbb{Z}}\delta(x-2\pi t)$.

In the rest of this section as well as Secs.~\ref{sec:position_rep} and \ref{sec:explicit_rel}, we set 
\begin{equation}
    \alpha=\sqrt{\frac{2\pi }{d}}\eqqcolon \alpha_d,
    \label{eq:square_shape} 
\end{equation}
which symmetrizes the code space in position and momentum coordinates in phase space \cite{Gottesman2001}.  This property of the code is meaningful even when the logical basis states are non-orthogonal, which is the case in approximate GKP codes.  In this paper, we adopt the following definition for this property.  
\begin{definition}[The code which is symmetric in position and momentum coordinates in phase space]\label{def:symmetric}
Let $\{\ket{j}: j=0,\ldots ,d-1\}$ be the logical qudit basis encoded in an oscillator mode.
The code is symmetric in position and momentum coordinates if it satisfies
\begin{equation}
    \begin{split}
&\mathrm{span}\{\ket{j}:j=0,\ldots,d-1\} \\
& \qquad \qquad =\mathrm{span}\{\hat{F}\ket{j}:j=0,\ldots,d-1\}.
    \end{split}
\end{equation} 
\end{definition} 
\noindent Note that we can use $\hat{F}^{\dagger}$ instead of $\hat{F}$ in the definition. 
The symmetric code is beneficial if we aim at unbiasing logical-level errors caused by physical-level phase-insensitive errors, that is, errors which occur symmetrically in position and momentum coordinates in phase space.  Furthermore, this definition implies that the Fourier transform $\hat{F}$ is an element of the stabilizer or a logical operator of the code since it preserves the code space.

The ideal GKP code can be regarded as a stabilizer code.  The stabilizer generators are given by the two commuting displacement operators ${X}_{\mathrm{st}}\coloneqq  \hat{X}(\alpha_d d)$ and ${Z}_{\mathrm{st}}\coloneqq  \hat{Z}(2\pi/\alpha_d)=\hat{Z}(\alpha_d d)$.  Similarly, logical Pauli operators can be defined as $X_{L}\coloneqq \hat{X}(\alpha_d)$ and $Z_{L}\coloneqq \hat{Z}(2\pi/(\alpha_d d))=\hat{Z}(\alpha_d)$, which satisfy $Z_{L}X_{L}=\exp(2\pi i/d)X_L Z_L$ as expected.
Using these stabilizer generators and logical Pauli operators, we have an alternative expression of the ideal GKP logical state as follows \cite{Gottesman2001,Albert2018}: 
\begin{align}
    &\ket{j^{(\mathrm{ideal})}} \nonumber \\
    &= \frac{(2d)^{-\frac{1}{4}}}{\vartheta(0,i d)} \sum_{\substack{s_1 \in \mathbb{Z}\\   s_2 \in \mathbb{Z}}} \hat{X}(\alpha_d(ds_1 + j)) \hat{Z}(\alpha_d s_2) \ket{0}_f \\
    &= \frac{(2d)^{-\frac{1}{4}}}{\vartheta(0,i d)} {X_L}^j \left(\sum_{l=0}^{d-1}{{Z}_L}^{l}\right)\sum_{\substack{s_1 \in \mathbb{Z}\\  s_2 \in \mathbb{Z}}}  {X_{\mathrm{st}}}^{s_1} {{Z}_{\mathrm{st}}}^{s_2} \ket{0}_f \\
    &\eqqcolon   {X_L}^j \left(\sum_{l=0}^{d-1}{Z_L}^{l}\right) P_{\mathrm{GKP}} \ket{0}_f,
    \label{eq:projection_to_logical_space} 
\end{align} 
where $\vartheta(0,i d)$ is the theta function, and the last line defines an operator $P_{\mathrm{GKP}}$, which is interpreted as the projection onto the code space ignoring the normalization. 
The consistency with Eq.~\eqref{eq:definition_ideal} can be confirmed as follows \cite{Albert2018}:
\begin{align}
    & \frac{(2d)^{-\frac{1}{4}}}{\vartheta(0,i d)} \sum_{\substack{s_1 \in \mathbb{Z}\\  s_2 \in \mathbb{Z}}} \hat{X}(\alpha_d(ds_1 + j)) \hat{Z}(\alpha_d s_2) \ket{0}_f \nonumber \\
    &= \frac{(2d)^{-\frac{1}{4}}}{\vartheta(0,i d)} \sum_{\substack{s_1 \in \mathbb{Z}\\  s_2 \in \mathbb{Z}}} \int \! dq\, e^{i q\alpha_d s_2} \ket{q+\alpha_d(ds_1 + j)}_{\hat{q}}\braket{q|0}_{f} \\ 
    &=  \frac{(2\pi d)^{-\frac{1}{4}}}{\vartheta(0,id)} \sum_{\substack{s_1 \in \mathbb{Z}\\  s_2 \in \mathbb{Z}}} \int \! dq \, e^{-\frac{1}{2}q^2 + i q\alpha_d s_2} \ket{q+\alpha_d(ds_1 + j)}_{\hat{q}} \\
    &= \frac{\sqrt{2\pi\alpha_d}}{\vartheta(0,i d)} \sum_{\substack{s_1 \in \mathbb{Z}\\  s'_2 \in \mathbb{Z}}}\int \!dq\, e^{-\frac{1}{2}q^2} \delta(q\alpha_d  + 2\pi s'_2 ) \ket{q+\alpha_d(ds_1 + j)}_{\hat{q}} \\
    &= \frac{\sqrt{\alpha_d d}}{ \vartheta(0,i d)} \sum_{\substack{s_1 \in \mathbb{Z}\\  s'_2 \in \mathbb{Z}}}  e^{-\pi d s'^2_2}  \Ket{q+\alpha_d\bigl(d(s_1-s'_2) + j\bigr)}_{\hat{q}} \\
    &= \sqrt{\alpha_d d} \sum_{s'_1\in \mathbb{Z}} \ket{q+\alpha_d (ds'_1 + j)}_{\hat{q}} \\
    & = \ket{j^{(\mathrm{ideal})}},
\end{align}
where we used ${}_{\hat{q}}\braket{q|0}_f=\pi^{-\frac{1}{4}}\exp(-q^2/2)$ in the second equality, used the Poisson summation formula $\sum_{s_2\in \mathbb{Z}}e^{-is_2 x} = 2\pi\sum_{s'_2\in\mathbb{Z}}\delta(x-2\pi s'_2)$ in the third equality, and defined $s'_1\coloneqq s_1-s'_2$ in the fifth equality.

In phase space, the Wigner function of the state $\ket{j^{(\mathrm{ideal})}}$ is given by \cite{Gottesman2001}
\begin{align}
&W_{\ket{j^{(\mathrm{ideal})}}\bra{j^{(\mathrm{ideal})}}}(q,p) \nonumber \\
&=\frac{1}{2}\sum_{\substack{t \in \mathbb{Z}\\  t' \in \mathbb{Z}}} e^{-\pi i tt'} \delta\biggl(p - \frac{\alpha_d t}{2} \biggr) \delta\biggl(q-\frac{\alpha_d dt'}{2}-\alpha_d j \biggr)  \label{eq:wigner_ideal_multi}\\
\begin{split} 
&= \frac{1}{2}\sum_{\substack{t \in \mathbb{Z}\\  t' \in \mathbb{Z}}} \delta\biggl(p + \frac{\alpha_d t}{2} \biggr) \left[\delta\biggl(q-\alpha_d d \left(t'+ \frac{j}{d}\right) \biggr) \right.\\
& \hspace{2cm}  \left. + (-1)^t\, \delta\biggl(q -\alpha_d d \left(t' + \frac{j}{d}+\frac{1}{2}\right)\biggr)\right].
\end{split}
\label{eq:wigner_ideal} 
\end{align}
This shows that the Wigner function of the ideal logical basis states forms a square lattice consisting of Dirac delta functions, which has half the period of the Dirac comb in the position and momentum representations.  Since its sublattice formed of the odd periods starting from $(q,p)=(\alpha_d j, 0)$ consists of the Dirac delta functions with negative signs, the comb at the odd periods in position cancel out when integrated over momentum, and vice versa. 

As defined so far, the ideal GKP code states are non-normalizable and thus unphysical.  Therefore, the ideal GKP code should be regarded as a limiting case of physically meaningful approximate codes.  Various approximations of the GKP code states are considered in the past literature \cite{Gottesman2001,Pirandola2004,Glancy2006,Vasconcelos2010,Menicucci2014,Albert2018,Noh2018,Weigand2018}.     
The following three approximations are conventionally used.

\vspace{0.3cm} 
\noindent (Approximation~1)
\begin{equation}
    \begin{split}
     \ket{j^{(1)}_{\kappa,\Delta}} &\coloneqq  \frac{1}{\sqrt{N_{\kappa,\Delta,j}^{(1)}}} \sum_{s\in\mathbb{Z}} e^{-\frac{1}{2} \kappa^2\alpha_d^2(ds+j)^2} \\
    & \hspace{2.2cm} \hat{X}(\alpha_d(ds+j))\hat{S}\left(-\ln \Delta\right)\ket{0}_f,
    \end{split}
    \label{eq:approx_1}
\end{equation}
where $\kappa, \Delta>0$, and ${N_{\kappa,\Delta,j}^{(1)}}$ is a normalization constant.  This approximate code state approaches the ideal one in the limit of $\kappa,\Delta\rightarrow 0$.  This approximation first appeared in the original paper of the GKP code \cite{Gottesman2001}.  The idea of this approximation is to replace the superposition of position ``eigenstates'' with that of squeezed coherent states with a squeezing parameter $\ln(1/\Delta)$, which are weighted by a Gaussian envelope of the width $1/\kappa$.  This gives us an insight about how to generate the GKP code state experimentally \cite{Motes2017}. 

\vspace{0.3cm} 
\noindent (Approximation~2) 
\begin{equation}
    \ket{j^{(2)}_{\gamma,\delta}}\coloneqq \frac{1}{\sqrt{ N_{\gamma,\delta,j}^{(2)}}}\iint \frac{dr_1 dr_2}{2\pi \gamma \delta}\  e^{-\frac{r_1^2}{2\gamma^2} - \frac{r_2^2}{2\delta^2}} \hat{V}(\bm{r}) \ket{j^{(\mathrm{ideal})}},
    \label{eq:approx_2}
\end{equation}
where $0<\gamma\delta<2$, and it approaches the ideal code state as $\gamma,\delta\rightarrow 0$.
This approximation also appeared in the original paper to regard the approximation as an error, and treat $\frac{1}{2\pi \gamma\delta}e^{-\frac{r_1^2}{2\gamma^2} - \frac{r_2^2}{2\delta^2}}$ as an error ``wave function'' \cite{Gottesman2001}.  They use the term ``wave function'' because the state given in Eq.~\eqref{eq:approx_2} is not an ideal code state subject to the error caused by the random displacement channel, but a coherent superposition of randomly displaced ideal code states.  The error ``wave function'' later turned out to have more profound meanings; it is actually a wave function in the ``grid representation'' \cite{Galetti1996,Ketterer2016,Terhal2016,Duivenvoorden2017,Weigand2018}, which is an analogous representation to the position representation, but with respect to the so-called ``shifted grid states'' instead of position eigenstates.  In Appendix \ref{sec:grid_representation}, we make remarks on the ``grid representation'' in terms of the representation theory of the Heisenberg group. 

\vspace{0.3cm} 
\noindent (Approximation~3)
\begin{equation}
    \ket{j^{(3)}_{\beta}} \coloneqq  \frac{1}{\sqrt{N_{\beta,j}^{(3)}}}e^{-{\beta \left(\hat{n}+\frac{1}{2}\right)}}\ket{j^{(\mathrm{ideal})}},
    \label{eq:approx_3}
\end{equation}
where $\beta$ satisfies $\beta>0$, and it approaches the ideal code state as $\beta\rightarrow 0$.  Contrary to the former two approximations, Approximation~3, first appearing in Ref.~\cite{Menicucci2014}, only deals with symmetric envelope in position and momentum coordinates.  
Since the approximation factor $e^{-\beta\left(\hat{n}+\frac{1}{2}\right)}$ is diagonal in the Fock basis, this approximation may be useful for computing the statistical properties of operators which are diagonal in the Fock basis, as shown in Ref.~\cite{Menicucci2014}.  On the other hand, though this approximate code state could conceptually be prepared by feeding the ideal code states to the beamsplitter followed by post-selecting the vacuum click at the idler port \cite{Noh2018}, it provides few implications about their realistic experimental generation.

\section{Equivalence of the approximations} \label{sec:main_result}
\subsection{Position and momentum representations} \label{sec:position_rep}
In order to determine the relationship among the three approximations, we derive the position and momentum representations, ${}_{\hat{q}}\braket{q|j}$ and ${}_{\hat{p}}\braket{p|j}$, of the approximate code states.
Note that the position and momentum representations of Approximation~1 have already appeared in the past literature \cite{Gottesman2001,Travaglione2002,Pirandola2004,Vasconcelos2010,Ketterer2016,Terhal2016,Motes2017,Douce2017,Weigand2018,Pantaleoni2019}, but we rewrite them for completeness.  For this purpose, we define the following functions.

\begin{definition}{\label{def:E_tilde_E}}
    Define $E_{\mu,\Gamma,a}(x)$ and $\tilde{E}_{\mu,\Gamma,a}(x)$ as
    \begin{align}
        E_{\mu,\Gamma,a}(x)&\coloneqq \exp\left(-\frac{x^2}{2\mu}\right) \sum_{s \in\mathbb{Z}} \delta\left(x-(s+a)\Gamma\right), \\
        \tilde{E}_{\mu,\Gamma,a}(x)&\coloneqq \exp\left(-\frac{x^2}{2\mu}\right) \sum_{s\in\mathbb{Z}} e^{2\pi i a s}\delta\left(x + s\Gamma\right) .
    \end{align}
\end{definition}
The function $E_{\mu,\Gamma,a}(x)$ is a Dirac comb with its interval given by $\Gamma$, which is shifted by the rational $a$ of the interval from the origin and weighted by the Gaussian $\exp(- x^2/(2\mu))$ of the width $\mu$.  It can also be interpreted as a Fourier transform of the theta function in the form of $\frac{1}{\sqrt{2\pi}}\vartheta \! \left[\begin{subarray}{c} a \\ \ \\ 0 \end{subarray} \right] \! \left(\frac{\Gamma}{2\pi} x, \frac{i \Gamma^2}{2\pi \mu}\right) $ with respect to $x$, which can be confirmed by its definition Eq.~\eqref{eq:theta_rational}.
On the other hand, the function $\tilde{E}_{\mu,\Gamma,a}(x)$, a Dirac comb with the Gaussian weight which has a phase factor for each peak, is a Fourier transform of the theta function in the form of $\frac{1}{\sqrt{2\pi}}\vartheta \! \left[\begin{subarray}{c} 0 \\ \ \\ a \end{subarray} \right] \! \left(-\frac{\Gamma}{2\pi} x, \frac{i \Gamma^2}{2\pi \mu}\right) $, which can also be confirmed by Eq.~\eqref{eq:theta_rational}.

Now, under Definition \ref{def:E_tilde_E}, we show the following proposition.
\begin{prop}[The position representation]\label{prop:position_rep}
    Let $\kappa,\Delta,\beta>0$ and $0< \gamma\delta < 2$.  Define $\lambda(\gamma,\delta)\coloneqq 1 + \frac{\gamma^2\delta^2}{4}$.
    Then, the position representations of the states Eqs.~\eqref{eq:approx_1}, \eqref{eq:approx_2}, and \eqref{eq:approx_3} are given as follows:
    \begin{itemize}
        \item (Approximation~1) 
        \begin{align}
    &{}_{\hat{q}}\braket{q|j^{(1)}_{\kappa,\Delta}} \nonumber \\
    &= \left(\frac{2\sqrt{\pi\Delta^2}}{N_{\kappa,\Delta,j}^{(1)}}\right)^{\frac{1}{2}}\; E_{\frac{1}{\kappa^2},\alpha_d d,\frac{j}{d}}*G_{\Delta^2}(q) \label{eq:convolution_rep_approx_1} \\
    \begin{split}
    &= \left(\frac{2\sqrt{\pi \Delta^2}}{\kappa^2 d N_{\kappa,\Delta,j}^{(1)}}\right)^{\frac{1}{2}} 
    G_{\frac{1 + \kappa^2\Delta^2}{\kappa^2}}(q) \\
    &\hspace{0.85cm} \times \vartheta \! \left[ \begin{subarray}{c} 0 \\ \ \\ j/d \end{subarray}\right] \! \left(-\frac{q}{\alpha_d d (1 + \kappa^2\Delta^2)},\frac{i\Delta^2}{d(1 + \kappa^2\Delta^2)}\right).
     \end{split} \label{eq:position_rep_approx_1} 
    \end{align}
    \item (Approximation~2)
    \begin{align}
    &{}_{\hat{q}}\braket{q|j^{(2)}_{\gamma,\delta}} \nonumber \\
    \begin{split}
    &=
    \left(\frac{\alpha_d d  }{ \lambda(\gamma,\delta) N_{\gamma,\delta,j}^{(2)}}\right)^{\frac{1}{2}} \\
    &\hspace{0.85cm} \times E_{\frac{\lambda(\gamma,\delta)}{\gamma^2}\left(1 - \frac{\gamma^2\delta^2}{2\lambda(\gamma,\delta)}\right)^{2},\, \alpha_d d \left(1 - \frac{\gamma^2\delta^2}{2\lambda(\gamma,\delta)}\right),\frac{j}{d}}*G_{\frac{\delta^2}{\lambda(\gamma,\delta)}}(q)
    \end{split} \label{eq:convolution_rep_approx_2} \\
    \begin{split}
    &=\left(\frac{\alpha_d \gamma^{-2}}{ N_{\gamma,\delta,j}^{(2)}}\right)^{\frac{1}{2}} G_{\frac{\lambda(\gamma,\delta)}{\gamma^2}}(q)\\ 
    & \hspace{0.85cm} \times \vartheta \! \left[ \begin{subarray}{c} 0 \\ \ \\ j/d \end{subarray}\right] \! \left(-\frac{q}{\alpha_d d}\left[1-\frac{\gamma^2 \delta^2}{2\lambda(\gamma,\delta)} \right], \frac{i \delta^2}{d \lambda(\gamma,\delta)}\right).
    \end{split}
    \label{eq:position_rep_approx_2} 
    \end{align}
    \item (Approximation~3)
    \begin{align}
    &{}_{\hat{q}}\braket{q|j^{(3)}_{\beta}} \nonumber \\
    &= \left(\frac{\alpha_d d}{\cosh\beta\, N_{\beta,j}^{(3)}}\right)^{\frac{1}{2}} \! E_{ \frac{1}{\sinh\beta\cosh\beta},\frac{\alpha_d d}{\cosh\beta},\frac{j}{d}}*G_{\tanh\beta}(q) \label{eq:convolution_rep_approx_3}  \\
    \begin{split}
    &= \left(\frac{\alpha_d}{\sinh\beta\, N_{\beta,j}^{(3)}}\right)^{\frac{1}{2}}\; G_{\frac{1}{\tanh\beta}}(q) \\
    & \hspace{2cm} \times \vartheta \! \left[\begin{subarray}{c} 0 \\ \ \\ j/d \end{subarray} \right] \! \left(-\frac{q}{\alpha_d d \cosh\beta},\frac{i \tanh\beta}{d}\right). 
    \end{split}
    \label{eq:position_rep_approx_3}
    \end{align} 
\end{itemize}
\end{prop}
For each approximation, we gave the two expressions in which we replace the Dirac delta functions in the definition of the ideal GKP code state with the Gaussian functions in different orders.  In the expressions \eqref{eq:convolution_rep_approx_1}, \eqref{eq:convolution_rep_approx_2}, and \eqref{eq:convolution_rep_approx_3}, each peak of the Dirac comb, which is weighted by a Gaussian as shown in the definition of $E_{\mu,\Gamma,a}$, is convoluted with another Gaussian $G_{\nu}(q)$.  In the alternative expressions \eqref{eq:position_rep_approx_1}, \eqref{eq:position_rep_approx_2}, and \eqref{eq:position_rep_approx_3}, the infinite sequence of Gaussian spikes as defined in $\vartheta\! \left[\begin{subarray}{c} 0 \\ \ \\ a \end{subarray} \right] \! (q,it)$ is multiplied by another Gaussian function $G_{\nu'}(q)$ which works as an overall envelope.
The expressions \eqref{eq:convolution_rep_approx_1}, \eqref{eq:convolution_rep_approx_2}, and \eqref{eq:convolution_rep_approx_3} are suited for understanding the physical structure of the approximation such as the interval of the neighboring Gaussian peaks.  The alternative expressions \eqref{eq:position_rep_approx_1}, \eqref{eq:position_rep_approx_2}, and \eqref{eq:position_rep_approx_3} are convenient for numerical calculations because algorithms to calculate the theta function with arbitrary precision are well known \cite{Deconinck2004}.

{\noindent \it Sketch of the proof.}
We derive Eqs.~\eqref{eq:convolution_rep_approx_1}, \eqref{eq:convolution_rep_approx_2}, and \eqref{eq:convolution_rep_approx_3} with straightforward but cumbersome calculations, and then apply the following lemma to derive Eqs.~\eqref{eq:position_rep_approx_1}, \eqref{eq:position_rep_approx_2}, and \eqref{eq:position_rep_approx_3}.

\begin{lemma} \label{lemma:conv_to_theta_func}
    For $\mu,\nu>0$, $\Gamma\in\mathbb{R}$, and $a\in\mathbb{Q}$, the following equality holds:
    \begin{equation}
        \begin{split}
        &E_{\mu,\Gamma,a}*G_{\nu}(q)\\
        &=\sqrt{\frac{2\pi\mu}{\Gamma^2}}\; G_{\mu+\nu}(q) \;\vartheta\! \left[\begin{subarray}{c} 0 \\ \ \\ a \end{subarray} \right] \! \left(-\frac{q}{(1+\nu/\mu)\Gamma},\frac{2\pi i \nu}{(1 + \nu/\mu)\Gamma^2}\right).
        \end{split}
    \end{equation}
\end{lemma}

{\noindent The full proof of Proposition \ref{prop:position_rep} as well as the proof of Lemma~\ref{lemma:conv_to_theta_func} is in Appendix \ref{sec:proof_position}. \qed}

Under Definition \ref{def:E_tilde_E}, the momentum representations of the approximate code states can also be given by the following corollary.
\begin{cor}[The momentum represenation] \label{cor:momentum_rep}
Let $\kappa,\Delta,\beta>0$ and $0<\gamma\delta<2$.  Let $\lambda(\gamma,\delta)\coloneqq 1 + \frac{\gamma^2\delta^2}{4}$.  Then, the momentum representations of the states \eqref{eq:approx_1}, \eqref{eq:approx_2}, and \eqref{eq:approx_3} are given as follows:
\begin{itemize}
    \item (Approximation~1)
    \begin{equation}
        \begin{split}
        {}_{\hat{p}}\braket{p|j^{(1)}_{\kappa,\Delta}} & =\left(\frac{2\sqrt{\pi \Delta^2}}{(1 + \kappa^2\Delta^2)  d N_{\kappa,\Delta,j}^{(1)}}\right)^{\frac{1}{2}} \\
        &\hspace{0.48cm} \times \tilde{E}_{\frac{1}{\Delta^2(1 + \kappa^2\Delta^2)},\frac{\alpha_d}{1 + \kappa^2\Delta^2},\frac{j}{d}} * G_{\frac{\kappa^2}{1 + \kappa^2\Delta^2}}(p).
        \end{split}
        \label{eq:momentum_rep_approx_1}
    \end{equation}
    \item (Approximation~2)
    \begin{equation}
        \begin{split}
        & {}_{\hat{p}}\braket{p|j^{(2)}_{\gamma,\delta}}\\
        &= \left(\frac{\alpha_d}{\lambda(\gamma,\delta) N_{\gamma,\delta,j}^{(2)}}\right)^{\frac{1}{2}} \\
        &\hspace{0.48cm} \times  \tilde{E}_{\frac{\lambda(\gamma,\delta)}{\delta^2}\left(1 - \frac{\gamma^2\delta^2}{2\lambda(\gamma,\delta)}\right)^{2},\, \alpha_d\left(1 - \frac{\gamma^2\delta^2}{2\lambda(\gamma,\delta)}\right),\frac{j}{d}} * G_{\frac{\gamma^2}{\lambda(\gamma,\delta)}} (p).
        \end{split}
        \label{eq:momentum_rep_approx_2}
    \end{equation}
    \item (Approximation~3)
    \begin{equation}
        \begin{split}
        &{}_{\hat{p}}\braket{p|j^{(3)}_{\beta}} \\
        &= \left(\frac{\alpha_d}{\cosh\beta\; N_{\beta,j}^{(3)}}\right)^{\frac{1}{2}}\! \tilde{E}_{\frac{1}{\sinh\beta\cosh\beta},\frac{\alpha_d}{\cosh\beta},\frac{j}{d}} * G_{\tanh\beta} (p).
        \end{split}
        \label{eq:momentum_rep_approx_3}
    \end{equation}
\end{itemize}
\end{cor}

\begin{proof}
We use the fact that the momentum representation of a state is a Fourier transform of its position representation, i.e., ${}_{\hat{p}}\braket{p|j}=\frac{1}{\sqrt{2\pi}}\int dq\, e^{-ipq}{}_{\hat{q}}\braket{q|j}$.
We can thus derive Eqs.~\eqref{eq:momentum_rep_approx_1}, \eqref{eq:momentum_rep_approx_2}, and \eqref{eq:momentum_rep_approx_3} as Fourier transforms of Eqs.~\eqref{eq:position_rep_approx_1}, \eqref{eq:position_rep_approx_2}, and \eqref{eq:position_rep_approx_3}, respectively, exploiting the fact that the Fourier transform of the product of two functions is given by the convolution of the Fourier transforms of the respective functions, and the Fourier transform of $\frac{1}{\sqrt{2\pi}} \vartheta \! \left[\begin{subarray}{c} 0 \\ \ \\ a \end{subarray} \right] \! \left(-\frac{\Gamma}{2\pi} x, \frac{i \Gamma^2}{2\pi \mu}\right) $ is $\tilde{E}_{\mu,\Gamma,a}$ while the Fourier transform of $G_{\nu}$ is $\sqrt{1/\nu}\, G_{\frac{1}{\nu}}$.
\end{proof}

\subsection{Explicit relations among the three approximations} \label{sec:explicit_rel}
The position and momentum representations of the three different approximate GKP code states lead to conditions for equivalence of these approximations.
Since $E_{\mu,\Gamma,a}*G_{\nu}(x)$ denotes the array of the Gaussian spikes $G_{\nu}(x)$ at intervals $\Gamma$, one can notice from Eqs.~\eqref{eq:convolution_rep_approx_2} and \eqref{eq:convolution_rep_approx_3} that the intervals of the Gaussian spikes of the approximate code states are narrower than those of the ideal one, $\alpha_d d$, in the case of Approximations~2 and 3.  
Furthermore, from Eqs.~\eqref{eq:momentum_rep_approx_2} and \eqref{eq:momentum_rep_approx_3}, the intervals of the Gaussian spikes of each of these approximate code states in the momentum representations get narrower in the same proportion as that of their respective position representations.
With this observation, Approximation~3, which has symmetric envelope functions in position and momentum representations, Eqs.~\eqref{eq:convolution_rep_approx_3} and \eqref{eq:momentum_rep_approx_3}, is expected to be a symmetric case ($\gamma=\delta$) of Approximation~2 in the sense of ``symmetric'' in Definition \ref{def:symmetric}. 
This can be confirmed by the following.
\begin{cor}[The symmetric code]\label{cor:symmetry_approx_3}
Let $\hat{F}$ be the Fourier operator defined in Sec.~\ref{sec:notation}.  Then, the following relation holds for the logical basis states of the Approximation~3:
\begin{equation}
\hat{F}^{\dagger}\ket{j^{(3)}_{\beta}}= \sum_{j'=0}^{d-1}\sqrt{\frac{N^{(3)}_{\beta,j'}}{N^{(3)}_{\beta,j}}}\ket{j'^{\, (3)}_{\beta}}.
\end{equation}
The same relation holds for Approximation~2 iff $\gamma=\delta$, i.e.,
\begin{equation}
    \hat{F}^{\dagger}\ket{j^{(2)}_{\gamma,\gamma}}= \sum_{j'=0}^{d-1}\sqrt{\frac{N^{(2)}_{\gamma,\gamma,j'}}{N^{(2)}_{\gamma,\gamma,j}}}\ket{j'^{\, (2)}_{\gamma,\gamma}}.
\end{equation}
\end{cor}
\begin{proof}
    It can be observed by combining ${}_{\hat{q}}\bra{x}\hat{F^{\dagger}} = {}_{\hat{p}}\bra{x}$ with Eqs.~\eqref{eq:convolution_rep_approx_2}, \eqref{eq:convolution_rep_approx_3}, \eqref{eq:position_rep_approx_2}, and \eqref{eq:momentum_rep_approx_3}.
\end{proof}  

In contrast with Approximations~2 and 3, the intervals of Gaussian spikes in the position representation \eqref{eq:convolution_rep_approx_1} of Approximation~1 are the same as those in the position representation of the ideal code state, and the intervals in the momentum representation \eqref{eq:momentum_rep_approx_1} of Approximation~1 are narrower than those in the momentum representation of the ideal code state; that is, Approximation~1 narrows the lattice spacing of the code space asymmetrically in position and momentum coordinates.
This suggests that Approximation~1 may be related to Approximation~2 or 3 by a transformation that symmetrizes the deviation of the lattice spacing in position and momentum coordinates.  

We confirm this by applying the squeezing operation $\hat{S}(\ln \sqrt{1 + \kappa^2\Delta^2})$ for symmetrizing the intervals of the Gaussian spikes of the code state $\ket{j^{(1)}_{\kappa,\Delta}}$ in position and momentum coordinates: 
\begin{align}
&{}_{\hat{q}}\bra{q}\hat{S}\left(\ln \sqrt{1 + \kappa^2\Delta^2}\right)\ket{j^{(1)}_{\kappa,\Delta}} \nonumber \\
&= (1 + \kappa^2\Delta^2)^{\frac{1}{4}} {}_{\hat{q}}\braket{\sqrt{1 + \kappa^2\Delta^2} q|j^{(1)}_{\kappa,\Delta}} \\
&= \sqrt{m}\ E_{\frac{1}{\kappa^2(1 + \kappa^2\Delta^2)},\frac{\alpha_d d}{\sqrt{1 + \kappa^2\Delta^2}},\frac{j}{d}}*G_{\frac{\Delta^2}{(1 + \kappa^2\Delta^2)}}(q), \label{eq:approx_1_squeezed} 
\end{align}
\begin{align}
&{}_{\hat{p}}\bra{p}\hat{S}\left(\ln \sqrt{1 + \kappa^2\Delta^2}\right)\ket{j^{(1)}_{\kappa,\Delta}} \nonumber \\
&= (1 + \kappa^2\Delta^2)^{-\frac{1}{4}} {}_{\hat{p}}\braket{p/\sqrt{1 + \kappa^2\Delta^2}|j^{(1)}_{\kappa,\Delta}} \qquad \\
&= \sqrt{\frac{m}{d}}\
\tilde{E}_{\frac{1}{\Delta^2},\frac{\alpha_d}{\sqrt{1 + \kappa^2\Delta^2}},\frac{j}{d}}*G_{\kappa^2}(p),
\label{eq:approx_1_squeezed_momentum}
\end{align} 
where $m=\frac{2}{N_{\kappa,\Delta,j}^{(1)}} \sqrt{\frac{\pi \Delta^2}{1 + \kappa^2\Delta^2}}$.  In order to derive Eqs.~\eqref{eq:approx_1_squeezed} and \eqref{eq:approx_1_squeezed_momentum}, we used $E_{\mu,\Gamma,a}*G_{\nu}(bx)=\frac{1}{b}E_{\frac{\mu}{b^2},\frac{\Gamma}{b},a}*G_{\frac{\nu}{b^2}}(x)$ and $\tilde{E}_{\mu,\Gamma,a}*G_{\nu}(bx)=\frac{1}{b}\tilde{E}_{\frac{\mu}{b^2},\frac{\Gamma}{b},a}*G_{\frac{\nu}{b^2}}(x)$, which can be obtained from the definition of the functions $E_{\mu,\Gamma,a}(x)$, $\tilde{E}_{\mu,\Gamma,a}(x)$, and $G_{\nu}(x)$.  Comparing the position representation \eqref{eq:approx_1_squeezed} of the squeezed version of Approximation~1 with the position representation \eqref{eq:convolution_rep_approx_2} of Approximation~2 and \eqref{eq:convolution_rep_approx_3} of Approximation~3, we arrive at the following theorem.

\begin{theorem}[Equivalence of the approximate GKP code states]\label{theorem:main}
By choosing the parameters in Approximations~1 and 2 as 
\begin{gather}
    \kappa^2 = \frac{\gamma^2}{\lambda(\gamma,\delta)}=\tanh\beta, \label{eq:set_Delta}\\ 
    \Delta^2 = \frac{\delta^2}{\lambda(\gamma,\delta)}\left(1 - \frac{\gamma^2\delta^2}{2\lambda(\gamma,\delta)}\right)^{-2} = \sinh\beta\cosh\beta, \label{eq:set_kappa}\\
    \gamma^2=\delta^2=2\tanh\frac{\beta}{2}, 
\end{gather}
where $\lambda(\gamma,\delta)\coloneqq 1 + \frac{\gamma^2\delta^2}{4}$, we have 
\begin{equation}
    \hat{S}\left(\ln \sqrt{1 + \kappa^2\Delta^2}\right)\ket{j^{(1)}_{\kappa,\Delta}} = \ket{j^{(2)}_{\gamma,\delta}}=\ket{j^{(3)}_{\beta}}.
\end{equation}
\end{theorem}

\begin{proof}
    It directly follows from Eqs.~\eqref{eq:convolution_rep_approx_2}, \eqref{eq:convolution_rep_approx_3}, and \eqref{eq:approx_1_squeezed}.
\end{proof}

Theorem~\ref{theorem:main} together with Corollary~\ref{cor:symmetry_approx_3} shows that up to a squeezing $\hat{S}\left(\ln \sqrt{1 + \kappa^2\Delta^2}\right)$ for Approximation~1 in order to make the code symmetric in the sense of Definition~\ref{def:symmetric}, the logical basis states of the symmetric code of Approximations~1, 2, and 3 are exactly the same quantum state.  This squeezing becomes negligible in the limit of good approximation.
In this sense, all these approximations are equivalent up to a squeezing that is ignorable in the limit of good approximation.  This definition of equivalence is well motivated since single-mode Gaussian unitary operations are easy to implement compared to non-Gaussian operations on CV systems such as optical systems, and among displacement, phase rotation, and squeezing for decomposing Gaussian operations \cite{Eisert2003}, only squeezing can change the lattice spacing.

The converse of the theorem is also true; the choice of parameters in Theorem~\ref{theorem:main} is the only choice for the logical basis states of these approximations to be the same quantum states.  This fact can be seen by the following remark.

\vspace{0.3cm} 
{\noindent \bf Remark 1:} 
So far, we followed the convention to fix the lattice spacing parameter as $\alpha = \alpha_d$, and derived equivalence relations among symmetric approximate codes.  Such an exact correspondence between approximate codes can be generalized to asymmetric case.  Let us remove the constraint of Eq.~\eqref{eq:square_shape} and regard $\alpha$ as a free parameter in each approximation, and define states $\ket{j^{(1)}_{\kappa,\Delta,\alpha}},\ \ket{j^{(2)}_{\gamma,\delta,\alpha}},$ and $\ket{j^{(3)}_{\beta,\alpha}}$ (see Appendix \ref{sec:proof_position}).  We can observe from Eqs.~\eqref{eq:prf_conv_ap1} and \eqref{eq:prf_conv_ap2} in Appendix \ref{sec:proof_position} that $\ket{j^{(1)}_{\kappa,\Delta,\alpha}}=\ket{j^{(2)}_{\gamma,\delta,\alpha'}}$ with the following choice of parameters:
\begin{gather}
    \kappa^2 = \frac{\gamma^2}{\lambda(\gamma,\delta)}\left(1 - \frac{\gamma^2\delta^2}{2\lambda(\gamma,\delta)}\right)^{-2}, \label{eq:corres_1}\\
    \alpha = \alpha' \left(1 - \frac{\gamma^2\delta^2}{2\lambda(\gamma,\delta)}\right), \label{eq:corres_2} \\
    \Delta^2 = \frac{\delta^2}{\lambda(\gamma,\delta)} \label{eq:corres_3}.
\end{gather}
Compared to $\ket{j^{(1)}_{\kappa,\Delta,\alpha}}$ and $\ket{j^{(2)}_{\gamma,\delta,\alpha}}$, the third approximation $\ket{j^{(3)}_{\beta,\alpha}}$ has fewer parameters and cannot always be made equivalent to $\ket{j^{(1)}_{\kappa,\Delta,\alpha}}$ and $\ket{j^{(2)}_{\gamma,\delta,\alpha}}$; that is, parameters in $\ket{j^{(1)}_{\kappa,\Delta,\alpha}}$ and $\ket{j^{(2)}_{\gamma,\delta,\alpha'}}$ are not redundant.  This is because each Gaussian spike of the third approximation $\ket{j^{(3)}_{\beta,\alpha}}$ always has the same variance in position and momentum.  Therefore, if we apply the squeezing $\hat{S}(\ln \zeta)$ to $\ket{j^{(3)}_{\beta,\alpha}}$ so that the variances of Gaussian spike in position and momentum can differ, we have $\ket{j^{(1)}_{\kappa,\Delta,\alpha}}=\ket{j^{(2)}_{\gamma,\delta,\alpha'}} =\hat{S}(\ln \zeta) \ket{j^{(3)}_{\beta,\alpha''}}$ with the following correspondence of the parameters in addition to Eqs.~\eqref{eq:corres_1}, \eqref{eq:corres_2}, and \eqref{eq:corres_3}:
\begin{gather}
    \kappa^2 = \zeta^2\sinh\beta\cosh\beta, \\
    \alpha = \frac{\alpha''}{\zeta\cosh\beta}, \\
    \Delta^2 = \frac{\tanh\beta}{\zeta^2}.
\end{gather}
This can be confirmed from the fact that ${}_{\hat{q}}\bra{q}\hat{S}(\ln\zeta)\ket{j}=\sqrt{\zeta}\, {}_{\hat{q}}\braket{\zeta q|j}$, and $E_{\mu,\Gamma,a}*G_{\nu}(\zeta q)=\zeta^{-1}E_{\frac{\mu}{\zeta^2},\frac{\Gamma}{\zeta},a}*G_{\frac{\nu}{\zeta^2}}(q)$.   

\vspace{0.3cm} 
{\noindent \bf Remark 2:}
The equivalence of Approximation 2 with $\gamma=\delta$ and 3 can also be proved from Eqs.~(1.4) and (7.12) in Ref.~\cite{Albert2018} by setting $l=l'=0$, while Ref.~\cite{Albert2018} does not prove the equivalence. Our contribution here is to derive their position wave functions in Proposition \ref{prop:position_rep} and to show the equivalence using these position wave functions.

\subsection{The standard form} \label{sec:standard_form}
Now that we have shown the equivalence of Approximations 1, 2, and 3, we introduce a standard form of the approximate GKP code state, which we will use in the rest of the paper.
\begin{definition}[Standard form of the approximate GKP code states]\label{def:standard_form}
    Given three parameters $\sigma_q^2,\ \sigma_p^2,$ and $\Gamma$, the standard form of the approximate GKP code is defined as the code which is spanned by a logical qudit basis $\{\ket{j_{\sigma_q^2,\sigma_p^2,\Gamma}}:j=0,\ldots,d-1\}$ with its position representation given by
    \begin{equation}
        \begin{split}
        &{}_{\hat{q}}\braket{q|j_{\sigma_q^2,\sigma_p^2,\Gamma}} \\
        &\coloneqq\left(\frac{2\Gamma \bigl(\Lambda(\sigma_q^2,\sigma_p^2)\bigr)^{-\frac{1}{2}}}{N_{\sigma_q^2,\sigma_p^2,\Gamma,j}}\right)^{\frac{1}{2}} E_{\frac{\Lambda(\sigma_q^2,\sigma_p^2)}{2\sigma_p^2}, \Gamma, \frac{j}{d}} * G_{2\sigma_q^2}(q), 
        \end{split} 
        \label{eq:standard_form}
    \end{equation}
    where $\Lambda(\sigma_q^2,\sigma_p^2) \coloneqq 1 - 4\sigma_q^2\sigma_p^2$, $0<\sigma_q^2 < 1/2$, $0 <\sigma_p^2 < 1/2$, and $N_{\sigma_q^2,\sigma_p^2,\Gamma,j}$ is a normalization constant.  For the symmetric code, the logical basis $\{\ket{j_{\sigma^2}}:j=0,\ldots,d-1\}$ is parametrized by only one parameter $\sigma^2$ $(0<\sigma^2<1/2)$ as
    \begin{equation}
        {}_{\hat{q}}\braket{q|j_{\sigma^2}}\coloneqq\left(\frac{2\alpha_d d}{N_{\sigma^2,j}}\right)^{\frac{1}{2}}\! E_{\frac{\Lambda(\sigma^2)}{2\sigma^2}, \alpha_d d\sqrt{\Lambda(\sigma^2)}, \frac{j}{d}} * G_{2\sigma^2}(q), \label{eq:standard_form_symmetric} 
    \end{equation}
    where $\Lambda(\sigma^2)\coloneqq 1-4\sigma^4$.
\end{definition}

{\noindent Note that $\ket{j_{\sigma^2}}$ is equal to $\ket{j_{\sigma^2,\sigma^2,\alpha_d d\sqrt{\Lambda(\sigma^2)}}}$.  The momentum representation of $\ket{j_{\sigma_q^2,\sigma_p^2,\Gamma}}$ is given by}
\begin{equation}
    \begin{split}
    &{}_{\hat{p}}\braket{p|j_{\sigma_q^2,\sigma_p^2,\Gamma}} \\
    &=\left(\frac{4\pi\sqrt{\Lambda(\sigma_q^2,\sigma_p^2)}}{\Gamma\, N_{\sigma_q^2,\sigma_p^2,\Gamma,j}}\right)^{\frac{1}{2}} \tilde{E}_{\frac{\Lambda(\sigma_q^2,\sigma_p^2)}{2\sigma_q^2}, \frac{2\pi\Lambda(\sigma_q^2,\sigma_p^2)}{\Gamma}, \frac{j}{d}} * G_{2\sigma_p^2}(p),
    \end{split}
    \label{eq:standard_form_p}
\end{equation}
and thus, for the symmetric code, it is given by
\begin{equation}
    {}_{\hat{p}}\braket{p|j_{\sigma^2}}=\left(\frac{2\alpha_d}{N_{\sigma^2,j}}\right)^{\frac{1}{2}} \tilde{E}_{\frac{\Lambda(\sigma^2)}{2\sigma^2}, \alpha_d \sqrt{\Lambda(\sigma^2)}, \frac{j}{d}} * G_{2\sigma^2}(p).
    \label{eq:symmetric_standard_p}
\end{equation}
We can also write Eqs.~\eqref{eq:standard_form}, \eqref{eq:standard_form_symmetric}, \eqref{eq:standard_form_p}, and \eqref{eq:symmetric_standard_p} in terms of the theta function by using Lemma~\ref{lemma:conv_to_theta_func}.

The physical meanings of the parameters $\sigma_q^2,\ \sigma_p^2,$ and $\Gamma$ of the state $\ket{j_{\sigma_q^2,\sigma_p^2,\Gamma}}$ (and hence $\sigma^2$ of the state $\ket{j_{\sigma^2}}$) will be clarified in Sec.~\ref{sec:wigner_rep}.  Furthermore, an explicit form of the normalization constant $N_{\sigma_q^2,\sigma_p^2,\Gamma,j}$ (and hence $N_{\sigma^2,j}$) is given in Proposition \ref{prop:normalization} in Sec.~\ref{sec:normalization_and_inner}.  
The representation corresponding to Eqs.~\eqref{eq:approx_1}, \eqref{eq:approx_2}, and \eqref{eq:approx_3} for the state $\ket{j_{\sigma_q^2,\sigma_p^2,\Gamma}}$ can be obtained by simply substituting the corresponding parameters.  For example, in the case of the representation corresponding to Approximation 1, we have from Eqs.~\eqref{eq:approx_1} and \eqref{eq:convolution_rep_approx_1} that 
\begin{equation}
    \begin{split}
    \ket{j_{\sigma_q^2,\sigma_p^2,\Gamma}} &= \frac{1}{\sqrt{N^{(1)}}} \sum_{s\in\mathbb{Z}} e^{- \frac{\sigma_p^2}{\Lambda(\sigma_q^2,\sigma_p^2)}\left(s+\frac{j}{d}\right)^2\Gamma^2} \\
    &\hspace{0.4cm}\times \hat{X}\left(\left(s+j/d\right)\Gamma\right)\hat{S}\left(-\ln \sqrt{2\sigma_q^2}\right)\ket{0}_f ,
    \end{split}
\end{equation}
where $N^{(1)}$ is given by
\begin{equation}
    N^{(1)} = \sqrt{2\pi \Gamma^{-2}\sigma_q^2\Lambda(\sigma_q^2,\sigma_p^2)}\; N_{\sigma_q^2,\sigma_p^2,\Gamma,j}.
\end{equation}

\begin{figure}[t]
    \centering
    \includegraphics[width=0.90\linewidth]{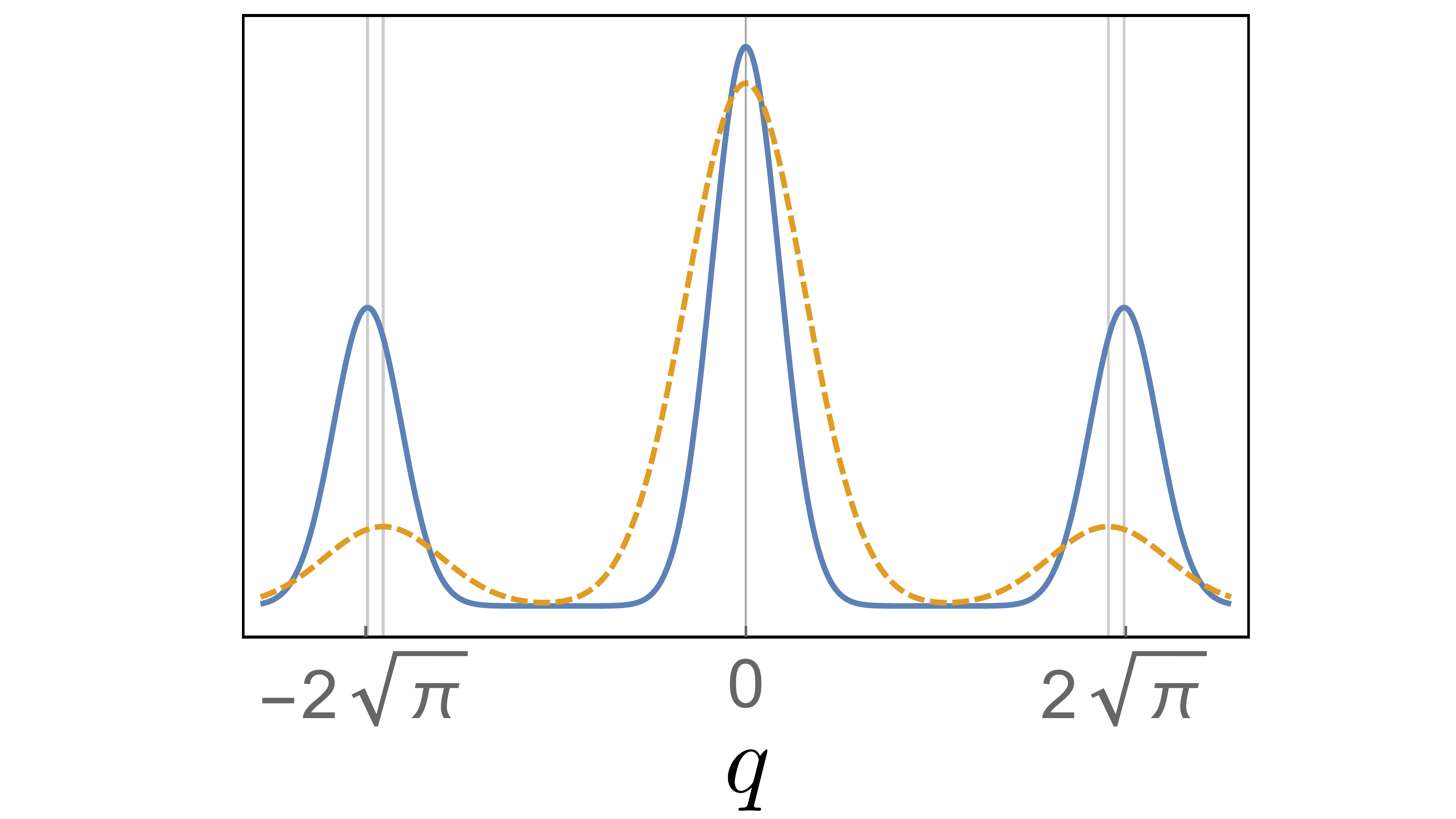}
    \caption{The position representation \eqref{eq:standard_form_symmetric} of the symmetric code state $\ket{0_{\sigma^2}}$ in the case $d=2$ with $\sigma^2=0.05$ (blue, solid line) and $\sigma^2=0.15$ (orange, dashed line).  The thin gray line shows the position of the peak of each Gaussian spike.  As shown in the main text, larger $\sigma^2$ leads to a larger change of the interval from $\alpha_d d$.  The change of the interval is $\mathcal{O}(\sigma^4)$, and thus may be negligible for small $\sigma^2$.  But our result enables quantitative analyses of the case in which $\sigma^2$ is not necessarily small, which is relevant to a current experimental technology \cite{Fluhmann2019,Campagne2019}.}
    \label{fig:position_standard_symmetric}
\end{figure}

\noindent Likewise, the representations of the symmetric code $\ket{j_{\sigma^2}}$ corresponding to Approximations 1 and 3 are given by
\begin{align}
    \begin{split}
    \ket{j_{\sigma^2}}  
    &= \frac{(d/\sigma^2)^{\frac{1}{4}} }{\sqrt{N_{\sigma^2,j}}} \sum_{s\in\mathbb{Z}} e^{-\sigma^2\alpha_d^2 \left(ds+j\right)^2} \\
    & \qquad \times \hat{X}\left(\alpha_d\left(d s+j\right)\sqrt{\Lambda(\sigma^2)}\right)\hat{S}\left(-\ln \sqrt{2\sigma^2}\right)\ket{0}_f 
    \end{split} \\ 
    &= \biggl(\frac{2}{\sqrt{\Lambda(\sigma^2)}N_{\sigma^2,j}}\biggr)^{\frac{1}{2}} e^{-\!\arctanh(2\sigma^2)\left(\hat{n} + \frac{1}{2}\right)}\ket{j^{(\mathrm{ideal})}}.
\end{align}

As shown in Sec.~\ref{sec:explicit_rel}, the position representation of the symmetric code state $\ket{j_{\sigma^2}}$ has narrower intervals of the neighboring Gaussian peaks than $\alpha_d d$ of the ideal one, which we illustrate in Fig.~\ref{fig:position_standard_symmetric}.  The change of the interval is $\mathcal{O}(\sigma^4)$, and thus may be negligible for small $\sigma^2$.  However, in experiment, we cannot always make $\sigma^2$ small enough to keep the change of the intervals negligible.  With our results, we can quantitatively analyze the code performance for any $\sigma^2$ that is not necessarily small.
\vspace{1cm}
\

\section{Explicit expressions of the Wigner function, inner products, and average photon number}\label{sec:applications}

In this section, we derive the expressions of the Wigner function, inner products, and the average photon number for the standard form of the approximate code state $\ket{j_{\sigma_q^2,\sigma_p^2,\Gamma}}$ in Definition \ref{def:standard_form}.  Those for $\ket{j_{\sigma^2}}$ can also be given by substituting $\sigma_q^2=\sigma_p^2=\sigma^2$ and $\Gamma=\alpha_d d\sqrt{\Lambda(\sigma^2)}$.  
They also have expressions in terms of the Riemann theta function \cite{Mumford2007} (also known as the Siegel theta function), which is a multivariable generalization of the theta function.  These alternative expressions are relatively neat and thus are given in Appendix \ref{sec:alternative_expression}.  In the main text, however, for ease of analyzing their asymptotic behaviors \eqref{eq:asymptotic_norm} and \eqref{eq:asymptotic_behavior}, we use the expressions in terms of the theta function.

\subsection{Wigner function} \label{sec:wigner_rep}

Here, we derive the Wigner function of the operators $\ket{j_{\sigma_q^2,\sigma_p^2,\Gamma}}\bra{j'_{\sigma_q^2,\sigma_p^2,\Gamma}}$.
The Wigner function of the approximate GKP code can be used for the analyses of quantum error correction as shown in Refs.~\cite{Menicucci2014,Fukui2017,Fukui2018,Fukui2019}.  
\begin{widetext}
\begin{prop}[Wigner function] \label{prop:wigner}
For the approximate code states $\ket{j_{\sigma_q^2,\sigma_p^2,\Gamma}}$ and $\ket{j'_{\sigma_q^2,\sigma_p^2,\Gamma}}$ in Definition \ref{def:standard_form}, the Wigner function $W_{\ket{j_{\sigma_q^2,\sigma_p^2,\Gamma}}\bra{j'_{\sigma_q^2,\sigma_p^2,\Gamma}}}(q,p)$ of the operator $\ket{j_{\sigma_q^2,\sigma_p^2,\Gamma}}\bra{j'_{\sigma_q^2,\sigma_p^2,\Gamma}}$ is given by
\begin{align}
&W_{\ket{j_{\sigma_q^2,\sigma_p^2,\Gamma}}\bra{j'_{\sigma_q^2,\sigma_p^2,\Gamma}}}(q,p)\nonumber \\
\begin{split}
&= \frac{1}{\sqrt{ N_{\sigma_q^2,\sigma_p^2,\Gamma,j}N_{\sigma_q^2,\sigma_p^2,\Gamma,j'}}}\left[ \left(E_{\frac{\Lambda(\sigma_q^2,\sigma_p^2)}{4\sigma_p^2},\Gamma,\frac{j+j'}{2d}}*G_{\sigma_q^2}(q) \right) \left(\tilde{E}_{\frac{\Lambda(\sigma_q^2,\sigma_p^2)}{4\sigma_q^2},\frac{\pi \Lambda(\sigma_q^2,\sigma_p^2)}{\Gamma},\frac{j-j'}{2d}} * G_{\sigma_p^2}(p)\right) \right. \\ 
& \hspace{4cm} \left. + \left( E_{\frac{\Lambda(\sigma_q^2,\sigma_p^2)}{4\sigma_p^2},\Gamma,\frac{j+j'}{2d} + \frac{1}{2}}*G_{\sigma_q^2}(q) \right) \left(\tilde{E}_{\frac{\Lambda(\sigma_q^2,\sigma_p^2)}{4\sigma_q^2},\frac{\pi \Lambda(\sigma_q^2,\sigma_p^2)}{\Gamma},\frac{j-j'}{2d} + \frac{1}{2}} * G_{\sigma_p^2}(p)\right) \right],
\end{split}\label{eq:wigner_func}
\end{align} 
where $E$ and $\tilde{E}$ are defined in Definition \ref{def:E_tilde_E}.
\end{prop}
\end{widetext}
{\noindent 
Calculations for deriving the Wigner function are similar to those for deriving the position and momentum representations, but are more complicated.  The proof of Proposition \ref{prop:wigner} is in Appendix~\ref{sec:proof_wigner}.  We can also write Eq.~\eqref{eq:wigner_func} in terms of the theta function by applying Lemma~\ref{lemma:conv_to_theta_func} to Eq.~\eqref{eq:wigner_theta_func} in Appendix~\ref{sec:proof_wigner}.}

The Wigner function in Proposition \ref{prop:wigner} shows the physical meanings of $\sigma_q^2,\ \sigma_p^2,$ and $\Gamma$.
The first term in the square bracket of Eq.~\eqref{eq:wigner_func} with $j=j'$ denotes an infinite sequence of Gaussian spikes each of which has variance $\sigma_q^2$ in position and $\sigma_p^2$ in momentum with periods $\Gamma$ and $\pi\Lambda(\sigma_q^2,\sigma_p^2)\Gamma^{-1}$, respectively, and has overall Gaussian envelopes with variances $(4\sigma_p^2)^{-1}\Lambda(\sigma_q^2,\sigma_p^2)$ and $(4\sigma_q^2)^{-1}\Lambda(\sigma_q^2,\sigma_p^2)$, respectively.  The second term shows that the same structure is also at the places shifted by half periods in position, but with positive and negative signs alternately in momentum.  The Gaussian spikes in the first and second terms with different signs interfere destructively when projected onto position or momentum, while constructively with the same signs.  Since $E_{\mu,\Gamma,a}(x)\rightarrow \sum_{s\in\mathbb{Z}}\delta\left(x- (s+a) \Gamma\right)$ and $\tilde{E}_{\mu,\Gamma,a}(x)\rightarrow \sum_{s\in\mathbb{Z}}e^{2\pi i a s}\delta\left(x+ s \Gamma\right)$ as $\mu\rightarrow \infty$, and $G_{\nu}(x)\rightarrow \delta(x)$ as $\nu\rightarrow 0$, we can observe that Eq.~\eqref{eq:wigner_func} with $\Gamma=\alpha_d d$ approaches Eq.~\eqref{eq:wigner_ideal} as $\sigma_q^2,\sigma_p^2\rightarrow 0$, as expected.

Using Eq.~\eqref{eq:wigner_func} with the explicit form of the normalization constant given in Sec.~\ref{sec:normalization_and_inner}, we plot the Wigner function of the GKP code state in Fig.~\ref{fig:wigner_fig}. 
Note that a similar expression has already been used in Ref.~\cite{Menicucci2014} with a more intuitive explanation.  Our contribution here is to derive the Wigner function corresponding to the approximate code states explicitly, which we will use in the detailed analysis of the average photon number.

\subsection{Normalization constant and inner product of the approximate code states} \label{sec:normalization_and_inner}
Using the Wigner function \eqref{eq:wigner_func}, we can provide a closed-form expression for $N_{\sigma_q^2,\sigma_p^2,\Gamma}$, while normalization constants were calculated numerically in previous works \cite{Menicucci2014,Terhal2016,Albert2018,Shi2019}.
Furthermore, since logical basis states of the approximate GKP codes are nonorthogonal, their inner products are nonzero in general, which we quantitatively analyze in the following.
Since the theta functions used in the following proposition can be calculated with arbitrary precision by a method in, e.g., Ref.~\cite{Deconinck2004}, the results are useful for evaluating the code performance reliably, as demonstrated in Ref.~\cite{Yamasaki2019}.

\begin{widetext}
\begin{prop}[Normalization constant and inner product] \label{prop:normalization}
The normalization factor $N_{\sigma_q^2,\sigma_p^2,\Gamma,j}$ of the approximate code state $\ket{j_{\sigma_q^2,\sigma_p^2,\Gamma}}$ in Definition \ref{def:standard_form} is given in terms of the theta functions by
\begin{equation}
    \begin{split}
    N_{\sigma_q^2,\sigma_p^2,\Gamma,j}  
    &= \vartheta \! \left[ \begin{subarray}{c} \frac{j}{d} \\ \ \\ 0 \end{subarray}\right] \! \left(0, 2\pi^{-1}i\Gamma^2\sigma_p^2 \bigl[\Lambda(\sigma_q^2,\sigma_p^2)\bigr]^{-1}\right)  \vartheta \! \left[ \begin{subarray}{c} 0 \\ \ \\ 0 \end{subarray}\right] \! \left(0,2\pi i \Gamma^{-2} \sigma_q^2 \Lambda(\sigma_q^2,\sigma_p^2) \right)  \\
    & \qquad  + \vartheta \! \left[ \begin{subarray}{c} \frac{j}{d} + \frac{1}{2} \\ \ \\ 0 \end{subarray}\right] \! \left(0,2\pi^{-1}i\Gamma^2\sigma_p^2 \bigl[\Lambda(\sigma_q^2,\sigma_p^2)\bigr]^{-1}\right)  \vartheta \! \left[ \begin{subarray}{c} 0 \\ \ \\ \frac{1}{2} \end{subarray}\right] \! \left(0,2\pi i \Gamma^{-2}\sigma_q^2 \Lambda(\sigma_q^2,\sigma_p^2)\right).
    \end{split} 
    \label{eq:normalization}
\end{equation} 
Furthermore, the inner product between $\ket{j_{\sigma_q^2,\sigma_p^2,\Gamma}}$ and the approximate code state $\ket{j'_{\sigma_q^2,\sigma_p^2,\Gamma}}$ is given by
\begin{align}
&\braket{j'_{\sigma_q^2,\sigma_p^2,\Gamma}|j_{\sigma_q^2,\sigma_p^2,\Gamma}} \nonumber \\
\begin{split}
 &= \frac{1}{\sqrt{N_{\sigma_q^2,\sigma_p^2,\Gamma,j} N_{\sigma_q^2,\sigma_p^2,\Gamma,j'}}}
 \biggl\{ \vartheta \! \left[ \begin{subarray}{c} \frac{j+j'}{2d} \\ \ \\ 0 \end{subarray}\right] \! \left(0,2\pi^{-1} i\Gamma^2\sigma_p^2 \bigl[\Lambda(\sigma_q^2,\sigma_p^2)\bigr]^{-1}\right)  \vartheta \! \left[ \begin{subarray}{c} 0 \\ \ \\ \frac{j-j'}{2d} \end{subarray}\right] \! \left(0,2\pi i \Gamma^{-2} \sigma_q^2 \Lambda(\sigma_q^2,\sigma_p^2) \right)  \\
&\hspace{4.5cm} + \vartheta \! \left[ \begin{subarray}{c} \frac{j+j'}{2d} + \frac{1}{2} \\ \ \\ 0 \end{subarray}\right] \! \left(0, 2\pi^{-1}i\Gamma^2\sigma_p^2 \bigl[\Lambda(\sigma_q^2,\sigma_p^2)\bigr]^{-1}\right) \vartheta \! \left[ \begin{subarray}{c} 0 \\ \ \\ \frac{j-j'}{2d} + \frac{1}{2} \end{subarray}\right] \! \left(0, 2\pi i \Gamma^{-2} \sigma_q^2 \Lambda(\sigma_q^2,\sigma_p^2) \right) \biggr\}.
\end{split} \label{eq:inner_product} 
\end{align} 
\end{prop}
\end{widetext}

\begin{figure}[t]
    \centering
    \includegraphics[width=0.85\linewidth]{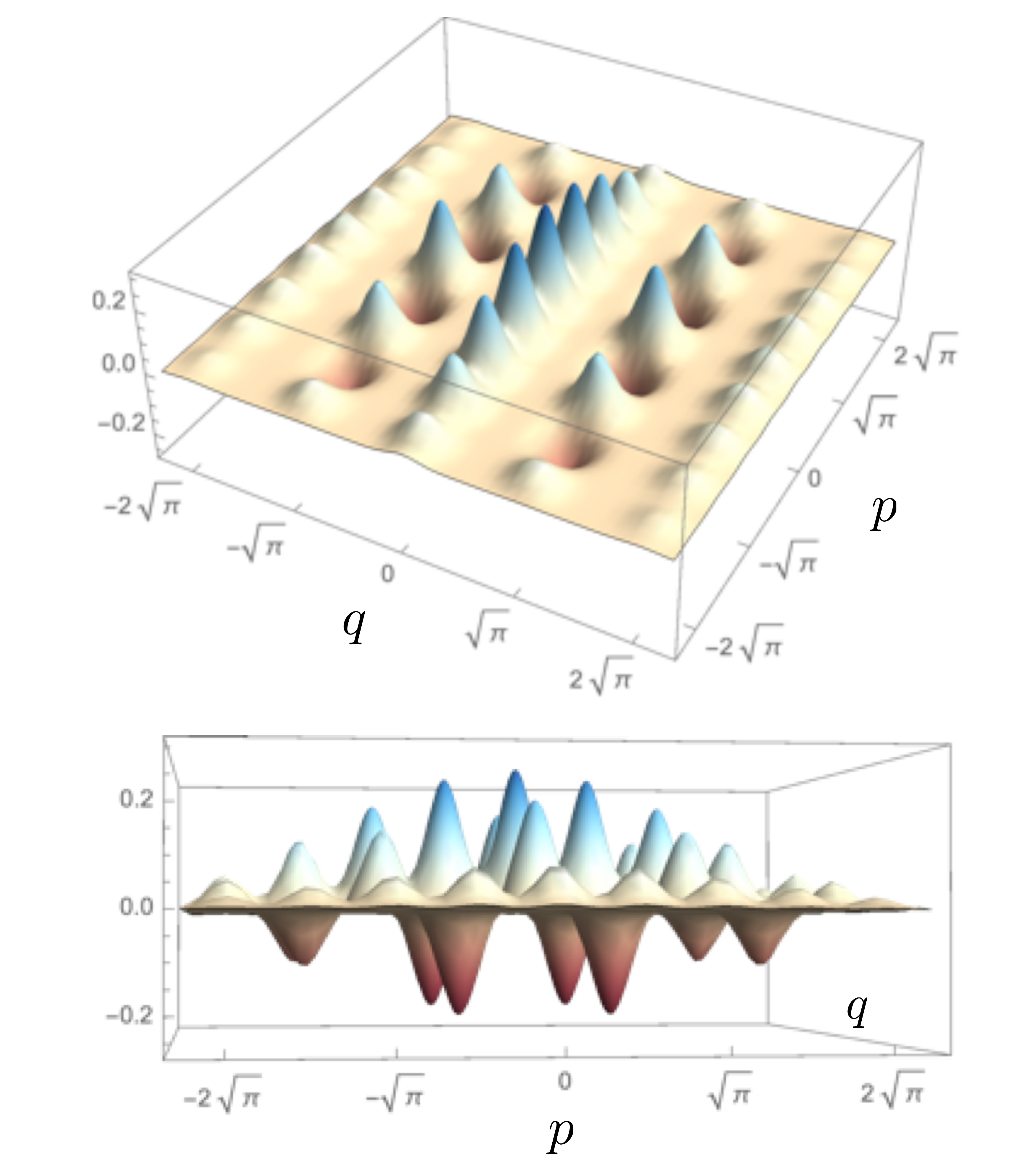}
    \caption{The Wigner function of the symmetric code state $\ket{0_{\sigma^2}}$ in the case $d=2$ with $\sigma^2=0.05$, which is calculated from Eq.~\eqref{eq:wigner_func} along with the explicit form of the normalization constant given in Sec.~\ref{sec:normalization_and_inner}.}
    \label{fig:wigner_fig}
\end{figure} 

\begin{proof}
We exploit the following facts:
\begin{gather}
\braket{j'|j}=\mathrm{Tr}\left[\ket{j}\!\bra{j'}\right]=\iint dq dp\ W_{\ket{j}\bra{j'}}(q,p), \label{eq:inner_trace}\\
\int dx\, f*g(x) = \int dx\, f(x) \int dy\, g(y), \label{eq:integral_convolution}\\
\begin{split}
\int dx\, E_{\mu,\Gamma,a}(x)&=\sum_{s\in\mathbb{Z}}\exp[-(s+a)^2\Gamma^2/2\mu] \\
&=\vartheta \! \left[\begin{subarray}{c} a \\ \ \\ 0 \end{subarray} \right]\!\left(0,\frac{i\Gamma^2}{2\pi\mu}\right), 
\end{split} \label{eq:integral_E} \\
\begin{split}
\int dx\, \tilde{E}_{\mu',\Gamma',a'}(x) &=\sum_{s\in\mathbb{Z}}\exp[-\Gamma'^2 s^2/2\mu'+2\pi i a' s] \\
&=\vartheta \! \left[\begin{subarray}{c} 0 \\ \ \\ a' \end{subarray} \right]\!\left(0,\frac{i\Gamma'^2}{2\pi\mu'}\right), 
\end{split} \label{eq:integral_tilde_E} \\
\int dx\; G_{\sigma^2}(x) = 1.  \label{eq:integral_normal} 
\end{gather}
Combining the above with the Wigner function of $W_{\ket{j_{\sigma_q^2,\sigma_p^2,\Gamma}}\bra{j'_{\sigma_q^2,\sigma_p^2,\Gamma}}}$ in Eq.~\eqref{eq:wigner_func}, we obtain Eqs.~\eqref{eq:normalization} and \eqref{eq:inner_product}.  
\end{proof}

The expressions in Proposition \ref{prop:normalization} are exact and applicable to any $\sigma^2$, but at the same time complicated.  Thus, we investigate their asymptotic behaviors in order to obtain intuitive relations with respect to the degree of approximation.  As shown in Ref.~\cite{Berndt2011}, the asymptotic behavior of the theta function in the form of $\vartheta \! \left[\begin{subarray}{c} a \\ \ \\ 0 \end{subarray} \right]\!\left(0,it\right)$ as $t\rightarrow +0$ is given by 
\begin{align}
    \vartheta \! \left[\begin{subarray}{c} a \\ \ \\ 0 \end{subarray} \right]\!\left(0,it\right) &= \sum_{s=0}^{\infty}e^{-\pi t (s+a)^2} + \sum_{s=0}^{\infty}e^{-\pi t (s + 1 - a)^2} \\
    &= \frac{1}{\sqrt{t}} + {\cal O}\left(t^{\frac{1}{2}}\right).
\end{align}
Furthermore, the asymptotic behavior of $\vartheta \! \left[\begin{subarray}{c} 0 \\ \ \\ a \end{subarray} \right]\!\left(0,it\right)$ as $t\rightarrow +0$ is given by
\begin{align} 
    \vartheta \! \left[\begin{subarray}{c} 0 \\ \ \\ a \end{subarray} \right] (0,it) &= \frac{1}{\sqrt{t}} \vartheta \! \left[\begin{subarray}{c} a \\ \ \\ 0 \end{subarray} \right]\! (0,it^{-1})\\
    & = \frac{1}{\sqrt{t}} \sum_{s=-\infty}^{\infty}e^{-\frac{\pi}{t} (s+a)^2} \\
    & \simeq \begin{cases} \frac{1}{\sqrt{t}}e^{-\frac{\pi}{t}  a^2} & \left(|a|\ll \frac{1}{2}\right) \\
             \frac{2}{\sqrt{t}}e^{-\frac{\pi}{t}  a^2} & \left(|a|\simeq \frac{1}{2}\right)
        \end{cases},
\end{align}
where we use Eq.~\eqref{eq:theta_func_identity_2} in Appendix \ref{sec:proof_position} in the first equality.
Now we derive the asymptotic form of the normalization constant $N_{\sigma_q^2,\sigma_p^2,\Gamma,j}$ in Eq.~\eqref{eq:normalization} as $\sigma_q^2,\sigma_p^2\rightarrow +0$,
\begin{equation}
    N_{\sigma_q^2,\sigma_p^2,\Gamma,j} \rightarrow \frac{1}{\sqrt{4\sigma_q^2\sigma_p^2}}.\label{eq:asymptotic_norm}
\end{equation}
In the same way, the asymptotic behavior of $\braket{j'_{\sigma_q^2,\sigma_p^2,\Gamma}|j_{\sigma_q^2,\sigma_p^2,\Gamma}}$ for $0\leq \frac{j'-j}{2d} \ll \frac{1}{2}$ in Eq.~\eqref{eq:inner_product} as $\sigma_q^2,\sigma_p^2\rightarrow +0$ is given by
\begin{equation}
    \braket{j'_{\sigma_q^2,\sigma_p^2,\Gamma}|j_{\sigma_q^2,\sigma_p^2,\Gamma}} 
    \rightarrow \exp\left[-\frac{(j'-j)^2\Gamma^2}{8d^2\sigma_q^2}\right]. \label{eq:asymptotic_behavior}
\end{equation}
The overlap between logical basis states thus decreases exponentially with respect to $\sigma_q^{-2}$.

\begin{figure}[t] 
    \centering
    \includegraphics[width=0.97\linewidth]{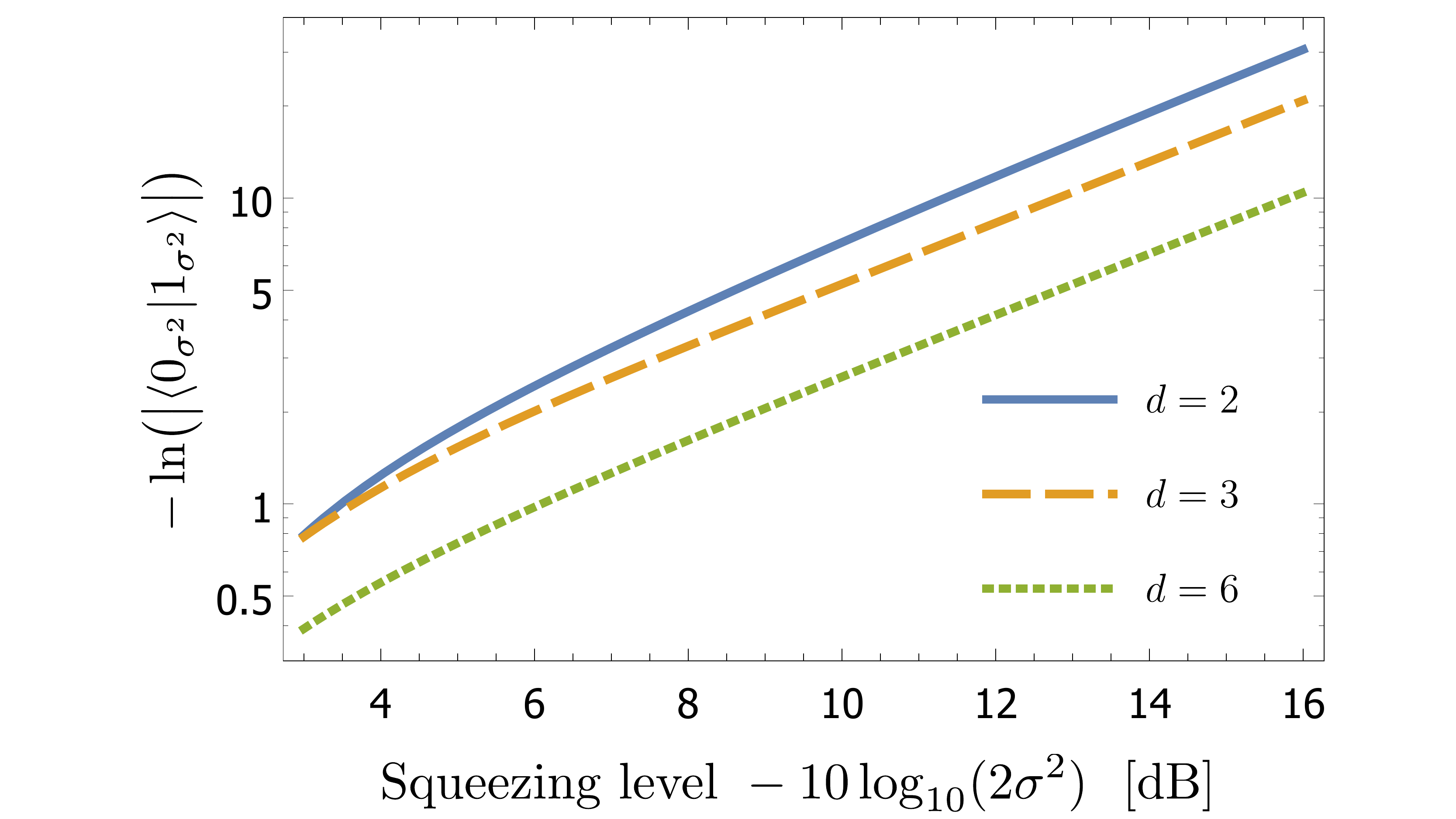}
    \caption{The logarithms of the absolute values of an inner product $-\ln\left(\bigl| \braket{0_{\sigma^2}|1_{\sigma^2}} \bigr| \right)$ for the code state \eqref{eq:standard_form_symmetric} in Definition \ref{def:standard_form} with $d=2,\ 3,$ and $6$.  The horizontal axis, $-10\log_{10}(2\sigma^2)$, is a squeezing level in decibels, which is a convention to express the degree of squeezing.  The vertical axis is in the log scale.  One can observe that, in the region where the squeezing level is over $5\ \si{dB}$ for $d=3$ and $6$, the negative logarithm of the inner product increases linearly with respect to the squeezing level in the log plot, that is, $-\ln\bigl| \braket{0_{\sigma^2}|1_{\sigma^2}}\bigr|\propto \sigma^{-2} $, as expected in the asymptotic behavior \eqref{eq:asymptotic_behavior}.} 
    \label{fig:graph_inner} 
\end{figure} 

Along with the asymptotic behaviors, we numerically calculate Eqs.~\eqref{eq:normalization} and \eqref{eq:inner_product} to see how the overlaps between code states change with respect to the degree of approximation.  Figure~\ref{fig:graph_inner} shows the logarithms of the absolute values of an inner product $\bigl|\braket{0_{\sigma^2}| 1_{\sigma^2}}\bigr|$ of the approximate code states \eqref{eq:standard_form_symmetric} in Definition \ref{def:standard_form} with $d=2,3,$ and $6$, with respect to a squeezing level in decibels $-10\log_{10}(2\sigma^2)$. 
One can observe that, in the region where the squeezing level is over $5\ \si{dB}$ for $d=3$ and $6$, the minus of the logarithm of the inner product increases linearly with respect to the squeezing level in the log plot, that is, $-\ln\bigl| \braket{0_{\sigma^2}|1_{\sigma^2}}\bigr|\propto \sigma^{-2} $, as expected in the asymptotic behavior \eqref{eq:asymptotic_behavior}.  In the case of $d=2$, the inclination of the plot is larger than those in the case of $d=3$ and $6$, which may be caused by a constant factor in Eq.~\eqref{eq:asymptotic_behavior} when $\frac{j+j'}{2d}\lesssim \frac{1}{2}$.
Note that the squeezing levels of the code states when $d=2$ in the recent experiments are $5.5 \text{--} 7.3\ \si{dB}$ with the position and momentum degrees of freedom in trapped ion system \cite{Fluhmann2019}, and $7.4 \text{--} 9.5 \ \si{dB}$ with the cavity mode of the superconducting system \cite{Campagne2019}.  The required squeezing level for the fault-tolerant threshold of the universal quantum computation is considered to be $8 \mathchar`- 16\ \si{dB}$ \cite{Fukui2018,Vuillot2019,Walshe2019,Fukui2019,Noh2019,Hanggli2020}, depending on experimental setups and noise models.

\subsection{Average photon number} \label{sec:ave_photon}
Using the Wigner function \eqref{eq:wigner_func} of the approximate code state $\ket{j_{\sigma_q^2,\sigma_p^2,\Gamma}}$, we can calculate the average photon number of the code state.  Below we write $\braket{\hat{A}}_{\ket{j_{\sigma_q^2,\sigma_p^2,\Gamma}}}\coloneqq \bra{j_{\sigma_q^2,\sigma_p^2,\Gamma}}\hat{A} \ket{j_{\sigma_q^2,\sigma_p^2,\Gamma}} $ for an operator $\hat{A}$.

\begin{widetext}
\begin{prop}[Average photon number] \label{prop:ave_photon}
    The average photon number $\braket{\hat{n}}_{\ket{j_{\sigma_q^2,\sigma_p^2,\Gamma}}}$ of the approximate code state $\ket{j_{\sigma_q^2,\sigma_p^2,\Gamma}}$ in Definition \ref{def:standard_form} is given as follows:
    \begin{equation}
        \braket{\hat{n}}_{\ket{j_{\sigma_q^2,\sigma_p^2,\Gamma}}} 
        = \frac{\sigma_q^2+\sigma_p^2 - 1}{2} 
         -\left(\frac{\partial}{\partial x} + \frac{\partial}{\partial y}\right) \ln \tilde{N}_{\sigma_q^2,\sigma_p^2,\Gamma,j}(x,y) \Biggr|_{x=4\sigma_p^2\bigl[\Lambda(\sigma_q^2,\sigma_p^2)\bigr]^{-1}\!,\, y=4\sigma_q^2\bigl[\Lambda(\sigma_q^2,\sigma_p^2)\bigr]^{-1}},  \label{eq:ave_photon_num} 
    \end{equation}
    where $\tilde{N}_{\sigma_q^2,\sigma_p^2,\Gamma,j}(x,y)$ is defined as
    \begin{equation}
    \tilde{N}_{\sigma_q^2,\sigma_p^2,\Gamma,j}(x,y) \coloneqq \vartheta \! \left[\begin{subarray}{c} \frac{j}{d} \\ \ \\ 0 \end{subarray} \right]\! \left(0,\frac{i\Gamma^2}{2\pi}x\right)\, \vartheta\!\left[\begin{subarray}{c} 0 \\ \ \\ 0 \end{subarray} \right]\! \left(0,\frac{\pi i \bigl[\Lambda(\sigma_q^2,\sigma_p^2)\bigr]^2}{2\Gamma^2} y \right) + \vartheta \! \left[\begin{subarray}{c} \frac{j}{d}+\frac{1}{2} \\ \ \\ 0 \end{subarray} \right]\! \left(0,\frac{i\Gamma^2}{2\pi}x\right)\, \vartheta\!\left[\begin{subarray}{c} 0 \\ \ \\ \frac{1}{2} \end{subarray} \right]\!\left(0, \frac{\pi i \bigl[\Lambda(\sigma_q^2,\sigma_p^2)\bigr]^2}{2\Gamma^2} y \right)\label{eq:tilde_N} 
    \end{equation}  
\end{prop}
\end{widetext}
{\noindent \it Sketch of proof.}  Using the Wigner function \eqref{eq:wigner_func}, we can derive the expectation values of the square of the position and momentum quadrature, $\braket{\hat{q}^2}_{\ket{j_{\sigma_q^2,\sigma_p^2,\Gamma}}}$ and $\braket{\hat{p}^2}_{\ket{j_{\sigma_q^2,\sigma_p^2,\Gamma}}}$.  Then we can derive $\braket{\hat{n}}_{\ket{j_{\sigma_q^2,\sigma_p^2,\Gamma}}}$ by exploiting the fact that $\braket{\hat{q}^2 + \hat{p}^2}_{\ket{j_{\sigma_q^2,\sigma_p^2,\Gamma}}} = \braket{2\hat{n} + 1}_{\ket{j_{\sigma_q^2,\sigma_p^2,\Gamma}}}$.  The full proof is in Appendix \ref{sec:expect_quad}. \qed

\subsection{The relation between squeezing level and average photon number} \label{sec:relation_sq_ave}
As an application of the results, we observe the relation between squeezing level and the average photon number of approximate code states.
The ``squeezing level'' of the GKP code state is a quality measure of an approximate code state.  It has a direct connection to the performance of the quantum error correction using GKP codes \cite{Menicucci2014,Fukui2017,Fukui2018,Vuillot2019,Walshe2019,Fukui2019,Hanggli2020}.
On the other hand, the average photon number of the encoded state is relevant to the capacity of the CV quantum channel \cite{Noh2018,Holevo1999,Wilde2018}, which works as an effective dimension of the Hilbert space.  Since it is found that the GKP code has high performance in the channel coding for bosonic Gaussian channels \cite{Albert2018,Noh2018}, the connections between these two notions are important for further analyses of the Gaussian channel coding.

``Squeezing level'' of the (symmetric) GKP code state was first considered in Ref.~\cite{Menicucci2014} in order to characterize the variance $\sigma^2$ of each convoluted Gaussian spike $G_{\sigma^2}$ in the Wigner function of the approximate code state, which directly affects the performance of the error correction with approximate GKP codes.  Since the squeezing level of a squeezed state is the logarithm of the ratio of the variances of the position quadrature $(\varDelta\hat{q})^2$ of that state and the vacuum state, Ref.~\cite{Menicucci2014} defines the squeezing level of the symmetric GKP code state by $-10\log_{10}(2\sigma^2)$ for the variance $\sigma^2$.  In the case of an asymmetric code state, there are two parameters $-10\log_{10}(2\sigma_q^2)$ and $-10\log_{10}(2\sigma_p^2)$, where $\sigma_q^2$ and $\sigma_p^2$ denote the variance of Gaussian spike in position and momentum, respectively, in the Wigner function of the standard form \eqref{eq:wigner_func}. 
Since the variance of Gaussian spike of the Wigner function of the code state in Approximation 1 is given by $\simeq\frac{\kappa^2}{2}$ when $\kappa=\Delta$ and $\kappa^2\Delta^2\ll 1$ as shown in Eq.~\eqref{eq:wigner_func}, the ``squeezing level'' is often identified with $-10\log_{10}\Delta^2$ $(\simeq -10\log_{10}\kappa^2)$ in Eq.~\eqref{eq:approx_1} \cite{Fukui2017,Fukui2018,Vuillot2019,Walshe2019,Campagne2019,Fukui2019}.
Note that there also exists another definition of ``effective squeezing parameter'', motivated by quantum metrology \cite{Duivenvoorden2017,Weigand2018,Fluhmann2019}. 
In this paper, we adopt the former definition as a ``squeezing level'' in order to observe the relation between the performance of error correction and the average photon number of the approximate code states.

Previous literature estimates the average photon number of the encoded state as $\simeq \frac{1}{4\sigma^2} - \frac{1}{2}$ for the symmetric code for given squeezing level $-10\log_{10}(2\sigma^2) \gg 1$ \cite{Gottesman2001,Glancy2006,Menicucci2014,Terhal2016,Noh2018,Albert2018}.  This is because the variance of the envelope Gaussian in the Wigner function of the approximate code states is roughly equal to $\frac{1}{4\sigma^2}$, and the average photon number relates to the expectation values of the squares of the position and momentum quadratures by $\braket{\hat{q}^2 + \hat{p}^2}_{\ket{j_{\sigma^2}}} = \braket{2\hat{n} + 1}_{\ket{j_{\sigma^2}}}$.  It is also consistent with the expression of the average photon number given in Eq.~\eqref{eq:ave_photon_num} when the asymptotic form $\tilde{N}_{\sigma_q^2,\sigma_p^2,\Gamma,j}(x,y)\propto \frac{1}{\sqrt{xy}}$ is considered.  
However, this estimation is no longer valid in the case of low squeezing level.  Here we are interested in the squeezing level at which this estimation deviates from the exact value.

\begin{figure}[t]  
            \centering
            \includegraphics[width=0.95\linewidth]{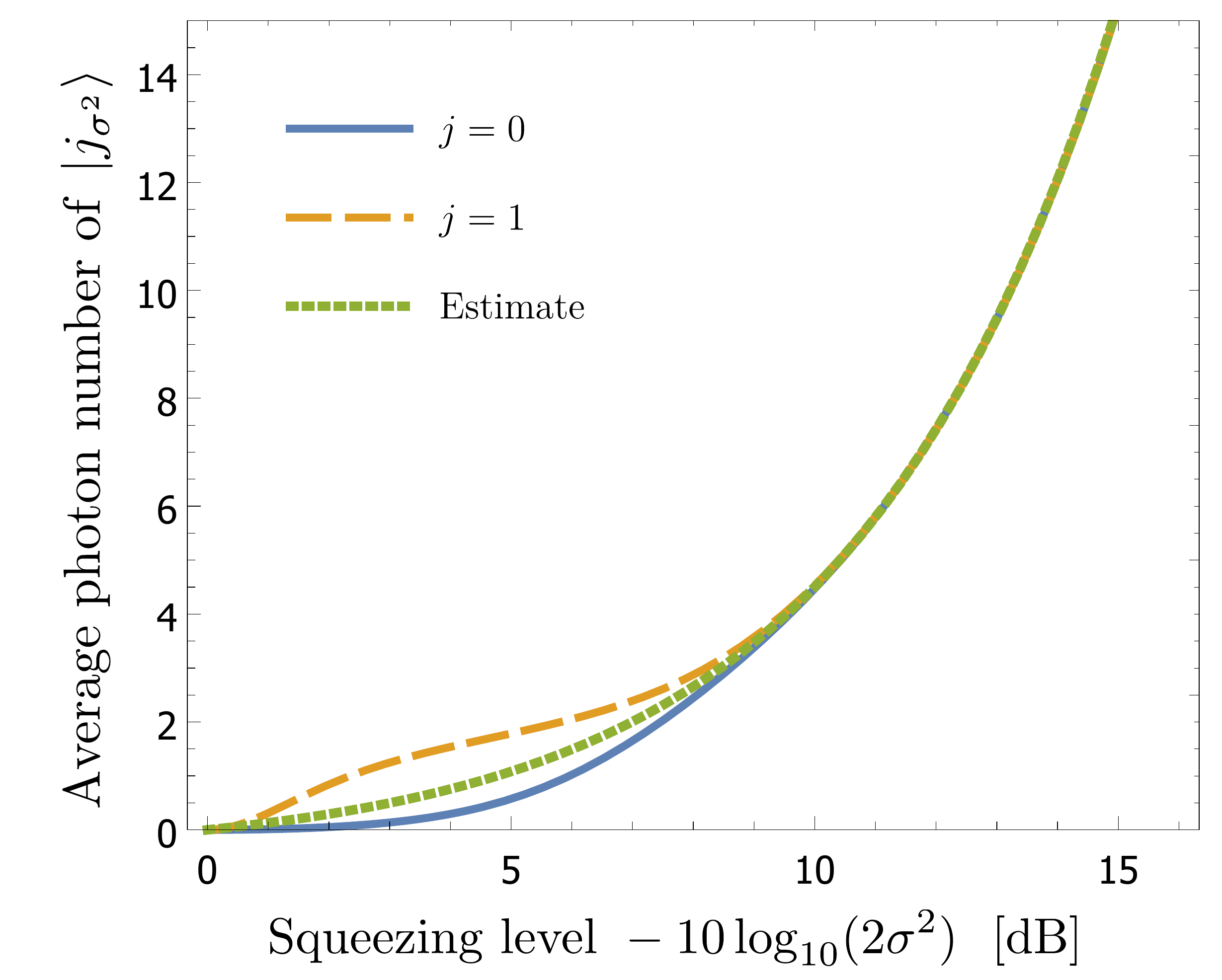}
            \caption{The average photon number of the code state \eqref{eq:standard_form_symmetric} in Definition \ref{def:standard_form} with $d=2$.  ``Estimate'' denotes the function $\frac{1}{4\sigma^2}-\frac{1}{2}$.  These three are in good accordance when the squeezing level is over $10\ \si{dB}$, but our rigorous calculations provide better estimates at the squeezing levels in the recent experiments, that is, $5.5 \text{--} 7.3\ \si{dB}$ in the trapped ion system \cite{Fluhmann2019} and $7.4 \text{--} 9.5 \ \si{dB}$ in the superconducting system \cite{Campagne2019}.}
    \label{fig:photon_num_d_2}
\end{figure} 
We compute the average photon number of the code state $\ket{j_{\sigma^2}}$ defined in Eq.~\eqref{eq:standard_form_symmetric} in Definition \ref{def:standard_form} with $d=2$, by using the formula \eqref{eq:ave_photon_num}.  As mentioned above, the squeezing level of $\ket{j_{\sigma^2}}$ is given by $-10\log_{10}(2\sigma^2)$.  Fig.~\ref{fig:photon_num_d_2} shows the average photon number of $\ket{0_{\sigma^2}}$ and $\ket{1_{\sigma^2}}$ with respect to the squeezing level $-10\log_{10}(2\sigma^2)$.
In Fig.~\ref{fig:photon_num_d_2}, we compare our result with a conventionally used estimate of the average photon number $\frac{1}{4\sigma^2}-\frac{1}{2}$.  The figure reveals that, when the squeezing level is less than $10\ \si{dB}$, the conventionally used estimate of the average photon number deviates from the exact values.  Note that $10\ \si{dB}$ squeezing is considered to be near a threshold for fault-tolerant CV quantum computation \cite{Fukui2018,Vuillot2019,Walshe2019,Fukui2019,Noh2019,Hanggli2020}, which is a curious coincidence.

\section{Conclusion} \label{sec:conclusion}
In this paper, we explicitly showed conditions under which the conventional approximations of the GKP code, Approximations~1, 2, and 3, defined as Eqs.~\eqref{eq:approx_1}, \eqref{eq:approx_2}, and \eqref{eq:approx_3}, are made equivalent.  We observed that up to a slight squeezing for Approximation~1, Approximations~1, 2, and 3 are equivalent for the symmetric code, in which the logical basis states and their Fourier transforms span the same code space.  Furthermore, we quantitatively showed that in all these approximations, the lattice spacing of the Gaussian spikes in phase space appearing in the description of the approximate code states is narrower than that of the corresponding ideal GKP code state.  
Although this effect may be negligible in the limit of large squeezing levels, it potentially affects the performance of error correction, since error correction strategy explicitly depends on the lattice spacing of the code states. 
Quantitatively, in the case of approximate code state of $d=2$ with $8\ \si{dB}$ squeezing, the lattice spacing is about $1 \%$ narrower than that of the ideal one.
It is thus needed to investigate error correction schemes taking the change in lattice spacing into account especially at a moderate squeezing level relevant to experimental realizations of GKP codes.

Exploiting the equivalence, we also gave the standard form of the approximate code states in terms of the position representation.  Furthermore, we derived the explicit formulas of the Wigner function, normalization constant, inner product, and the average photon number of the logical basis states.  We hope that these tools given in the present paper accelerate further theoretical developments of CV quantum information processing based on quantum error correction and channel coding with the GKP error-correcting code.

\acknowledgements
The authors thank K.\ Fukui, K.\ Maeda, Y.\ Kuramochi, and T.\ Sasaki for the helpful discussion.  This work was supported by CREST (Japan Science and Technology Agency) JPMJCR1671 and Cross-ministerial Strategic Innovation Promotion Program (SIP) (Council
for Science, Technology and Innovation (CSTI)).

\appendix
\section{The grid representation} \label{sec:grid_representation}
The grid representation appeared in the paper by Zak \cite{Zak1968}, and was later elaborated upon \cite{Janssen1982,Galetti1996} and used in the context of quantum information theory \cite{Ketterer2016,Terhal2016,Duivenvoorden2017,Weigand2018}.  We review it here.
Let $(u,v)\in[0,1)\times [0,1)$, and $\hat{\mathcal{V}}(u,v)\coloneqq \hat{V}\left((2\pi v/(\alpha_d d),\alpha_d d u)^{\top}\right)$.  Then, $e^{-\pi it\hat{I}}\hat{\mathcal{V}}$ forms a Heisenberg group $e^{-\pi it\hat{I}}\hat{\mathcal{V}}(u,v)\cdot e^{-\pi it'\hat{I}}\hat{\mathcal{V}}(u',v')=e^{-\pi i(t+t'+uv'-u'v)}\hat{\mathcal{V}}(u+u',v+v')$.  
Define $\ket{u,v}_{\mathrm{grid}}$ as 
\begin{align}
\ket{u,v}_{\mathrm{grid}} &\coloneqq \hat{\mathcal{V}}(u,v) \ket{0^{(\mathrm{ideal})}} \\
&=e^{-\pi i uv}\hat{Z}(2\pi v/(\alpha_d d))\hat{X}(\alpha_d d u) \ket{0^{(\mathrm{ideal})}}.
\end{align}
In Refs.~\cite{Terhal2016,Duivenvoorden2017,Weigand2018}, $\ket{u,v}_{\mathrm{grid}}$ with $d=1$ is called the ``shifted grid state''. 
The generalized ``shifted grid state'' $\ket{u,v}_{\mathrm{grid}}$ with arbitrary $d$ satisfies an orthogonality and completeness relation in the following sense \cite{Ketterer2016,Weigand2018}:
\begin{gather*}
{}_{\mathrm{grid}}\braket{u,v|u',v'}_{\mathrm{grid}} = \delta(u-u')\delta(v-v'), \\
\int_{0}^{1} du \int_{0}^{1} dv\ \ket{u,v}_{\mathrm{grid}} \bra{u,v} = \hat{I}.
\end{gather*}
The ``wave function'' $\phi_f(u,v)$ of a state $\ket{f}$ with respect to the ``shifted grid states'', i.e., the grid representation of $\ket{f}$, is defined as $\phi_f(u,v)\coloneqq {}_{\mathrm{grid}}\braket{u,v|f}$, which satisfies
\begin{equation}
\int_{0}^{1} du \int_{0}^{1} dv\ \bigl|\phi_f(u,v)\bigr|^2 = 1. 
\label{eq:square_integrability}
\end{equation}
The ``wave function'' of the ideal GKP logical basis state $\ket{j^{(\mathrm{ideal})}}$ can be regarded as a Dirac delta function centered at $(j/d,0)$, which does not satisfy Eq.~\eqref{eq:square_integrability} and therefore, cannot be regarded as a physical state.  However, functions satisfying Eq.~\eqref{eq:square_integrability} and localized at $(j/d,0)$ are well-defined approximate logical basis states.

Given the position representation $\psi_{f}(q)\coloneqq {}_{\hat{q}}\braket{q|f}$ of a (pure) state $\ket{f}$, its grid representation $\phi_{f}(u,v)$ can be given by 
\begin{align}
\phi_f(u,v)&\coloneqq {}_{\mathrm{grid}}\braket{u,v|f}  \\
&= \int dq\ {}_{\mathrm{grid}}\braket{u,v|q}_{\hat{q}}\braket{q|f}  \\
&= \sqrt{\alpha_d d} \sum_{s\in\mathbb{Z}} e^{-2\pi i v \left(s+\frac{u}{2}\right)}\psi_f\left(\alpha_d d(u+s)\right).
\label{eq:basis_transf_pos}
\end{align}  
Using the last equality, we can expand the domain $[0,1)\times[0,1)$ of the ``wave function'' of the grid representation $\phi$ to $\mathbb{R}^2$.  This redefined ``wave function'' $\phi:\mathbb{R}^2\rightarrow \mathbb{C}$ satisfies Eq.~\eqref{eq:square_integrability} and the following:
\begin{equation} 
    \begin{split}
        &\forall (n_1,n_2)^{\top}\in\mathbb{Z}^2, \\
        &\quad \phi(u+n_1,v+n_2)=e^{-\pi i (n_1n_2+u n_2 - vn_1)}\phi(u,v),
    \end{split}
\label{eq:periodic_cond}
\end{equation}
which can be confirmed from Eq.~\eqref{eq:basis_transf_pos}.  The functions $\phi:\mathbb{R}^2\rightarrow\mathbb{C}$ which satisfy Eqs.~\eqref{eq:square_integrability} and \eqref{eq:periodic_cond} form a representation space of the Heisenberg group called $L^2(\mathbb{R}^2\sslash\mathbb{Z}^2)$ \cite{Mumford2007_2}, where the action of the group element $\mathrm{Op}(\cdot)$ on $\phi$ is given by
\begin{align}
&\mathrm{Op}(e^{-\pi i t\hat{I}}\hat{\mathcal{V}}(u,v))\phi_f(x,y) \nonumber \\
&\coloneqq {}_{\mathrm{grid}}\bra{x,y}e^{-\pi i t\hat{I}}\hat{\mathcal{V}}(u,v)\ket{f} \\
&= e^{-\pi i (t + xv-yu)}\phi_f(x-u,y-v). 
\end{align}
The formulation can easily be generalized to the $g$-mode case by considering the representation space $L^2(\mathbb{R}^{2g}\sslash\mathbb{Z}^{2g})$ \cite{Mumford2007_2}.

\section{The proof of Proposition \ref{prop:position_rep} and Lemma~\ref{lemma:conv_to_theta_func}} \label{sec:proof_position}
First, we derive Eqs.~\eqref{eq:convolution_rep_approx_1}, \eqref{eq:convolution_rep_approx_2}, and \eqref{eq:convolution_rep_approx_3} in Proposition \ref{prop:position_rep}.
In the main text, $\alpha$ is fixed to $\alpha_d$ for $\ket{j^{(\mathrm{ideal})}}$ and all the approximations, but here, for later use, we perform calculation for a general $\alpha$, that is, derive the position representation of $\ket{j^{(1)}_{\kappa,\Delta,\alpha}},\ \ket{j^{(2)}_{\gamma,\delta,\alpha}},$ and $\ket{j^{(3)}_{\beta,\alpha}}$.
We start with the derivation of Eq.~\eqref{eq:convolution_rep_approx_1}.  We have
\begin{align}
    &{}_{\hat{q}}\braket{q|j^{(1)}_{\kappa,\Delta,\alpha}}\nonumber \\
    \begin{split}
    &= \frac{1}{\sqrt{N_{\kappa,\Delta,j}^{(1)}}}\sum_{s\in\mathbb{Z}} e^{-\frac{1}{2} \kappa^2\alpha^2(ds+j)^2}  \\
    & \hspace{2.2cm} \times {}_{\hat{q}}\bra{q}\hat{X}(\alpha(ds+j)) \hat{S}\left(-\ln \Delta\right)\ket{0}_f 
    \end{split} \\
    &=  \sum_{s\in\mathbb{Z}}  \frac{e^{-\frac{1}{2} \kappa^2\alpha^2(ds+j)^2}}{\sqrt{\Delta N_{\kappa,\Delta,j}^{(1)}}} {}_{\hat{q}}\braket{\left(q-\alpha(ds+j)\right)\!/\!\Delta | 0}_f \\
    \begin{split}
    &= \left(\sqrt{\pi \Delta^2} N_{\kappa,\Delta,j}^{(1)} \right)^{-\frac{1}{2}} \\
    &\hspace{1.3cm}\times \sum_{s\in\mathbb{Z}}e^{-\frac{1}{2}\kappa^2 \alpha^2 d^2 \left(s+\frac{j}{d}\right)^2 -\frac{1}{2\Delta^2}\left(q - \alpha d \left(s+\frac{j}{d}\right)\right)^2}  
    \end{split} \label{eq:middle_stage_approx_1} \\
    &= \left(\frac{2\sqrt{\pi\Delta^2}}{N_{\kappa,\Delta,j}^{(1)}}\right)^{\frac{1}{2}} E_{\frac{1}{\kappa^2},\alpha d,\frac{j}{d}}*G_{\Delta^2}(q), \label{eq:prf_conv_ap1}
\end{align}
where we used $X(a)\ket{q}_{\hat{q}}=\ket{q + a}_{\hat{q}}$ and $S(r)\ket{q}_{\hat{q}} = e^{-r/2}\ket{e^{-r}q}_{\hat{q}}$ in the second equality, and ${}_{\hat{q}}\braket{q|0}_f=\pi^{-\frac{1}{4}}\exp(-q^2/2) $ in the third equality.  Substituting $\alpha$ with $\alpha_d$ in Eq.~\eqref{eq:prf_conv_ap1}, we obtain Eq.~\eqref{eq:convolution_rep_approx_1}.

The derivation of Eq.~\eqref{eq:convolution_rep_approx_2} is similar.  We have
\begin{align}
    &{}_{\hat{q}}\braket{q|j^{(2)}_{\gamma,\delta,\alpha}} \nonumber \\
    &=\frac{1}{\sqrt{N_{\gamma,\delta,j}^{(2)}}}\iint \frac{dr_1 dr_2}{2\pi \gamma \delta}\,  e^{-\frac{r_1^2}{2\gamma^2} - \frac{r_2^2}{2\delta^2}} \bra{q}\hat{V}(\bm{r}) \ket{j^{(\mathrm{ideal})}} \\
    \begin{split}
    &= \frac{1}{\sqrt{ N_{\gamma,\delta,j}^{(2)}}}\iint \frac{dr_1 dr_2}{2\pi \gamma\delta}\, e^{-\frac{r_1^2}{2\gamma^2} - \frac{r_2^2}{2\delta^2}-\frac{i r_1 r_2}{2} + ir_1 q} \\
    & \hspace{4cm} \times {}_{\hat{q}}\braket{q - r_2|j^{(\mathrm{ideal})}},
    \end{split} \label{eq:middle_ap2}
\end{align}
where we used $\hat{V}(\bm{r})\coloneqq \exp(-ir_p r_q /2)\hat{Z}(r_p)\hat{X}(r_q)$, $\hat{Z}(r_p)\ket{q}_{\hat{q}}=e^{ir_p q}$, and $\hat{X}(r_q)\ket{q}_{\hat{q}}=\ket{q+r_q}$.
Using ${}_{\hat{q}}\braket{q - r_2|j^{(\mathrm{ideal})}} = \sum_{s\in\mathbb{Z}}\delta\left(\alpha (d s + j) - q+r_2\right)$, we have
\begin{align}
    &\eqref{eq:middle_ap2} \nonumber \\
    \begin{split}
    &= \biggl(\frac{\alpha d }{ N_{\gamma,\delta,j}^{(2)}}\biggr)^{\frac{1}{2}} \iint \frac{dr_1 dr_2}{2\pi \gamma\delta}\, e^{-\frac{r_1^2}{2\gamma^2} - \frac{r_2^2}{2\delta^2}-\frac{i r_1 r_2}{2} + ir_1 q} \\
    & \hspace{3.2cm} \times \sum_{s\in\mathbb{Z}} \delta\left(r_2 - q + \alpha(ds + j)\right) 
    \end{split} \\
    \begin{split}
    &= \biggl(\frac{\alpha d }{ N_{\gamma,\delta,j}^{(2)}}\biggr)^{\frac{1}{2}} \sum_{s\in\mathbb{Z}} \int \frac{dr_1}{2\pi \gamma\delta}\, e^{-\frac{1}{2\gamma^2}\left[r_1 - \frac{i\gamma^2}{2}(q + \alpha(ds + j))\right]^2 } \\
    & \hspace{2.5cm} \times e^{-\frac{\gamma^2}{8}\left(q+\alpha(ds + j)\right)^2 - \frac{1}{2\delta^2} \left(q-\alpha(ds + j)\right)^2 } 
    \end{split} \\
    \begin{split}
    &= \left(\frac{\alpha d }{2\pi \delta^2 N_{\gamma,\delta,j}^{(2)}}\right)^{\frac{1}{2}} e^{-\frac{\lambda(\gamma,\delta) q^2}{2\delta^2}} \\
    &\hspace{1cm} \times \sum_{s\in\mathbb{Z}} e^{- \frac{\alpha^2 d^2 \lambda(\gamma,\delta)}{2\delta^2}\left(s + \frac{j}{d}\right)^2 + \frac{\alpha d q}{\delta^2}\left(\lambda(\gamma,\delta) -\frac{\gamma^2\delta^2}{2}\right)\left(s + \frac{j}{d}\right)}
    \end{split} \\ 
    \begin{split}   
    &= \left(\frac{\alpha d }{2\pi \delta^2 N_{\gamma,\delta,j}^{(2)}}\right)^{\frac{1}{2}}\!\sum_{s\in\mathbb{Z}} e^{-\frac{\lambda(\gamma,\delta)}{2\delta^2}\left[q - \alpha d\left(1-\frac{\gamma^2\delta^2}{2\lambda(\gamma,\delta)}\right)\left(s + \frac{j}{d}\right)\right]^2 } \\ 
    &\hspace{2.5cm} \times e^{-\frac{\alpha^2d^2 \lambda(\gamma,\delta)}{2\delta^2}\bigl[1 - \left(1-\frac{\gamma^2\delta^2}{2\lambda(\gamma,\delta)}\right)^2\bigr]\left(s + \frac{j}{d}\right)^2 } 
    \end{split} \\
    \begin{split}
    &= \left(\frac{\alpha d  }{\lambda(\gamma,\delta)  N_{\gamma,\delta,j}^{(2)}}\right)^{\frac{1}{2}}\\
    &\hspace{0.9cm} \times E_{\frac{\lambda(\gamma,\delta)}{\gamma^2}\left(1-\frac{\gamma^2\delta^2}{2\lambda(\gamma,\delta)}\right)^{2},\, \alpha d \left(1-\frac{\gamma^2\delta^2}{2\lambda(\gamma,\delta)}\right),\frac{j}{d}}*G_{\frac{\delta^2}{\lambda(\gamma,\delta)}}(q), 
    \end{split} \label{eq:prf_conv_ap2}
\end{align} 
where we used a Gaussian integral in the third equality, and used
\begin{equation}
    1 - \left(1-\frac{\gamma^2\delta^2}{2\lambda(\gamma,\delta)}\right)^2 = \left(\frac{\gamma\delta}{\lambda(\gamma,\delta)}\right)^2
\end{equation}
in the last equality.
Substituting $\alpha$ with $\alpha_d$ in Eq.~\eqref{eq:prf_conv_ap2} leads to Eq.~\eqref{eq:convolution_rep_approx_2}.

The derivation of Eq.~\eqref{eq:convolution_rep_approx_3} needs a trick.  We have
\begin{align}
    &{}_{\hat{q}}\braket{q|j^{(3)}_{\beta,\alpha}} \nonumber \\ 
    &= \frac{1}{\sqrt{N_{\beta,j}^{(3)}}} {}_{\hat{q}}\bra{q}\sum_{n\in\mathbb{N}}\ket{n}\bra{n}_f e^{-\beta\left(n+\frac{1}{2}\right)}\ket{j^{(\mathrm{ideal})}}  \\
    &= \left(\frac{\alpha d}{N_{\beta,j}^{(3)}}\right)^{\frac{1}{2}}\sum_{s\in\mathbb{Z}}\sum_{n\in\mathbb{N}}e^{-\beta\left(n+\frac{1}{2}\right)}\psi_n(q)\psi_n^{*}(\alpha (ds + j)) ,
\end{align}
where $\psi_n(x)\coloneqq (2^n n!\sqrt{\pi})^{-1/2}e^{-x^2/2}H_n(x)$ denotes the wave function of the Fock state.  Using Mehler's Hermite polynomial formula \cite{Weisstein} 
\begin{equation}
    \begin{split}
    &\sum_{n\in \mathbb{N}}\frac{(u/2)^n}{n!}H_n(x)H_n(y)\exp\left(-\frac{x^2 + y^2}{2}\right) \\
    & = \frac{1}{\sqrt{1 - u^2}} \exp\left[- \frac{(1 + u^2)(x^2+y^2) - 4uxy}{2(1 - u^2)}\right],
    \end{split}
\end{equation}
we obtain
\begin{align}
    &{}_{\hat{q}}\braket{q|j^{(3)}_{\beta}} \nonumber \\
    \begin{split}
    &= \left(\frac{\pi^{-1} e^{-\beta} \alpha d }{(1-e^{-2\beta})N_{\beta,j}^{(3)}}\right)^{\frac{1}{2}} \\
    & \hspace{0.5cm} \times \sum_{s\in\mathbb{Z}}\exp\left[-\frac{(1+e^{-2\beta})\left(q^2+\alpha^2(ds+j)^2\right)}{2(1-e^{-2\beta})}\right] \\
    & \hspace{3.3cm} \times \exp\left[-\frac{4e^{-\beta} \alpha(ds + j)q}{2(1-e^{-2\beta})}\right] 
    \end{split} \\
    \begin{split}
    &= \left(\frac{(2\pi)^{-1}\alpha d}{\sinh\beta N_{\beta,j}^{(3)}}\right)^{\frac{1}{2}}  \\
    & \hspace{1.2cm} \times \sum_{s\in\mathbb{Z}}e^{-\frac{\alpha^2 d^2 }{2\tanh\beta}\left(s+\frac{j}{d}\right)^2 + \frac{\alpha d q }{\sinh\beta}\left(s+\frac{j}{d}\right) - \frac{q^2}{2\tanh\beta}}
    \end{split} \\
    & \nonumber \\
    \begin{split}
    &= \left(\frac{(2\pi)^{-1}\alpha d}{\sinh\beta N_{\beta,j}^{(3)}}\right)^{\frac{1}{2}} \\
    &\hspace{1cm} \times \sum_{s\in\mathbb{Z}} e^{-\frac{1}{2\tanh\beta}\left(q - \frac{\alpha d}{\cosh\beta}\left(s+\frac{j}{d}\right) \right)^2 -\frac{\alpha^2 d^2\tanh\beta}{2}\left(s + \frac{j}{d}\right)^2 } 
    \end{split} \\
    &= \left(\frac{\alpha d}{\cosh\beta\; N_{\beta,j}^{(3)}}\right)^{\frac{1}{2}}\ E_{ \frac{1}{\sinh\beta\cosh\beta},\frac{\alpha d}{\cosh\beta},\frac{j}{d}}*G_{\tanh\beta}(q). \label{eq:prf_conv_ap3}
\end{align}
Substituting $\alpha$ with $\alpha_d$ in Eq.~\eqref{eq:prf_conv_ap3} leads to Eq.~\eqref{eq:convolution_rep_approx_3}.

Next, we prove Lemma~\ref{lemma:conv_to_theta_func} to derive Eqs.~\eqref{eq:position_rep_approx_1}, \eqref{eq:position_rep_approx_2}, and \eqref{eq:position_rep_approx_3} from Eqs.~\eqref{eq:convolution_rep_approx_1}, \eqref{eq:convolution_rep_approx_2}, and \eqref{eq:convolution_rep_approx_3}, respectively.
From the definition of $E_{\mu,\Gamma,a}$ in Definition \ref{def:E_tilde_E} as well as the definition of $G_{\nu}$ in Eq.~\eqref{eq:gaussian}, we have
\begin{align}
    &E_{\mu,\Gamma,a}*G_{\nu}(q) \nonumber \\
    &= \frac{1}{\sqrt{2\pi \nu}} \sum_{s\in\mathbb{Z}}\exp\biggl[-\frac{ (s+a)^2\Gamma^2}{2\mu} - \frac{\left(q-(s+a)\Gamma\right)^2}{2\nu}\biggr]  \\
    &= \frac{e^{-\frac{1}{2\nu}q^2}}{\sqrt{2\pi \nu}}\, \vartheta \! \left[\begin{subarray}{c} a \\ \ \\ 0 \end{subarray} \right] \! \left(\frac{\Gamma q}{2\pi i \nu}, \frac{i(1+\nu/\mu)\Gamma^2}{2\pi \nu}\right). \label{eq:midle_conv}
\end{align}
The theta function has the following identity \cite{Mumford2007}
\begin{equation}
    \vartheta (z/\tau,-1/\tau) = (-i\tau)^{\frac{1}{2}}\exp(\pi i z^2/\tau)\, \vartheta(z,\tau),
\end{equation}
which leads to
\begin{equation}
    \vartheta \! \left[\begin{subarray}{c} 0 \\ \ \\ a \end{subarray} \right] \! (z/\tau,-1/\tau) = (-i\tau)^{\frac{1}{2}} \exp(\pi i z^2/\tau)\, \vartheta \! \left[\begin{subarray}{c} a \\ \ \\ 0 \end{subarray} \right] \! (z,\tau). \label{eq:theta_func_identity_2}
\end{equation}
Applying this to Eq.~\eqref{eq:midle_conv}, we have
\begin{align}
    &\eqref{eq:midle_conv} \nonumber \\
    &= \frac{e^{\left(-\frac{1}{2\nu}+\frac{1}{2\nu(1+\nu/\mu)}\right)q^2}}{\sqrt{(1+\nu/\mu)\Gamma^2}} \vartheta \! \left[\begin{subarray}{c} 0 \\ \ \\ a \end{subarray} \right] \! \left(-\frac{q}{(1+\nu/\mu)\Gamma}, \frac{2\pi i \nu}{(1+\nu/\mu)\Gamma^2}\right) \\
    &= \sqrt{\frac{2\pi\mu}{\Gamma^2}}\; G_{\mu+\nu}(q)\; \vartheta \! \left[\begin{subarray}{c} 0 \\ \ \\ a \end{subarray} \right] \! \left(-\frac{q}{(1+\nu/\mu)\Gamma}, \frac{2\pi i \nu}{(1+\nu/\mu)\Gamma^2}\right),
\end{align}
which proves Lemma~\ref{lemma:conv_to_theta_func}.  Then, as mentioned above, we obtain Eqs.~\eqref{eq:position_rep_approx_1}, \eqref{eq:position_rep_approx_2}, and \eqref{eq:position_rep_approx_3} by applying Lemma~\ref{lemma:conv_to_theta_func} to Eqs.~\eqref{eq:convolution_rep_approx_1}, \eqref{eq:convolution_rep_approx_2}, and \eqref{eq:convolution_rep_approx_3}, respectively. \qed

\begin{widetext}
\section{Proof of Proposition \ref{prop:wigner}} \label{sec:proof_wigner}
We compute $W_{\ket{j_{\sigma_q^2,\sigma_p^2,\Gamma}}\bra{j'_{\sigma_q^2,\sigma_p^2,\Gamma}}}$ as follows:
\begin{align}
& W_{\ket{j_{\sigma_q^2,\sigma_p^2,\Gamma}}\bra{j'_{\sigma_q^2,\sigma_p^2,\Gamma}}} \nonumber \\
&= \frac{1}{\pi}\int dx\ e^{2i px} {}_{\hat{q}}\braket{q-x|j_{\sigma_q^2,\sigma_p^2,\Gamma}}\braket{j'_{\sigma_q^2,\sigma_p^2,\Gamma}|q+x}_{\hat{q}} \\
&=\frac{2 \Gamma \bigl(\Lambda(\sigma_q^2,\sigma_p^2)\bigr)^{-\frac{1}{2}}}{\pi\sqrt{N_{\sigma_q^2,\sigma_p^2,\Gamma,j}N_{\sigma_q^2,\sigma_p^2,\Gamma,j'}}} \int dx\ e^{2i px} \left(E_{\frac{\Lambda(\sigma_q^2,\sigma_p^2)}{2\sigma_p^2},\Gamma,\frac{j}{d}}*G_{2\sigma_q^2}(q-x)\right)\left(E_{\frac{\Lambda(\sigma_q^2,\sigma_p^2)}{2\sigma_p^2},\Gamma,\frac{j'}{d}} * G_{2\sigma_q^2}(q+x)\right)
\\
\begin{split}
&=\frac{ \left(2\pi^2 \sigma_q^2\sqrt{\Lambda(\sigma_q^2,\sigma_p^2)}\right)^{-1} \Gamma}{ \sqrt{N_{\sigma_q^2,\sigma_p^2,\Gamma,j}N_{\sigma_q^2,\sigma_p^2,\Gamma,j'}}} \int dx\, e^{2i px}  \sum_{s} \exp\left[-\frac{\left(s+ \frac{j}{d}\right)^2\Gamma^2\sigma_p^2}{\Lambda(\sigma_q^2,\sigma_p^2)}-\frac{1}{4\sigma_q^2}\left(q-x-\left(s+\frac{j}{d}\right)\Gamma\right)^2\right]  \\
&\hspace{6cm} \times \sum_{s'}\exp\left[-\frac{\left(s'+ \frac{j'}{d}\right)^2\Gamma^2\sigma_p^2}{\Lambda(\sigma_q^2,\sigma_p^2)}-\frac{1}{4\sigma_q^2}\left(q+x-\left(s'+\frac{j'}{d}\right)\Gamma\right)^2\right]
\end{split}\\
\begin{split}
&= \frac{ \left(2\pi^2 \sigma_q^2\sqrt{\Lambda(\sigma_q^2,\sigma_p^2)}\right)^{-1} \Gamma}{\sqrt{N_{\sigma_q^2,\sigma_p^2,\Gamma,j}N_{\sigma_q^2,\sigma_p^2,\Gamma,j'}}} \int dx\, \sum_{s,s'}\exp\left( -\frac{1}{2\sigma_q^2}\left\{x-i\left[2\sigma_q^2 p + \frac{i \Gamma}{2}\left( s+\frac{j}{d} -s'-\frac{j'}{d}\right) \right]\right\}^2\right) \\
& \hspace{3cm} \times \exp\left\{ -\frac{1}{2\sigma_q^2}\left[2\sigma_q^2 p + \frac{i \Gamma}{2}\left( s+\frac{j}{d} -s'-\frac{j'}{d}\right) \right]^2  - \frac{1}{2\sigma_q^2} \left[q^2 - \Gamma q \left(s + \frac{j}{d} + s' + \frac{j'}{d} \right)\right] \right\}  \\
&\hspace{4cm} \times \exp\left\{-\frac{\Gamma^2}{2} \left(\frac{\sigma_p^2}{\Lambda(\sigma_q^2,\sigma_p^2)} + \frac{1}{4\sigma_q^2}\right)\left[\left(s + \frac{j}{d} + s' + \frac{j'}{d}\right)^2 + \left(s + \frac{j}{d} -  s' - \frac{j'}{d}\right)^2\right] \right\} \end{split} \\
\begin{split}
&= \frac{\bigl(2\pi^3\sigma_q^2\Lambda(\sigma_q^2,\sigma_p^2)\bigr)^{-\frac{1}{2}} \Gamma}{\sqrt{N_{\sigma_q^2,\sigma_p^2,\Gamma,j}N_{\sigma_q^2,\sigma_p^2,\Gamma,j'}}} \sum_{s,s'} \exp\left\{-\frac{\Gamma^2\sigma_p^2}{2 \Lambda(\sigma_q^2,\sigma_p^2)} \left[\left(s + \frac{j}{d} + s' + \frac{j'}{d}\right)^2 + \left(s + \frac{j}{d} -  s' - \frac{j'}{d}\right)^2 \right] \right\}  \\
& \hspace{4.5cm} \times \exp\left\{- \frac{1}{2\sigma_q^2} \left[q - \frac{\Gamma}{2} \left(s + \frac{j}{d} + s' + \frac{j'}{d} \right)\right]^2 - 2\sigma_q^2 p^2 - i\Gamma p\left( s+\frac{j}{d} -s'-\frac{j'}{d} \right) \right\}  
\end{split}
\label{eq:before_changing_var}
\end{align}
where we used the standard form \eqref{eq:standard_form} in the second equality.  At this stage, we will change the variables for the summation from $s$ and $s'$ to $s+s'$ and $s - s'$.  Since $s+s'$ and $s-s'$ have the same parity, the summation splits into two parts: one with $s+s'=2t,\ s-s'=2t'$, ($t,t'\in\mathbb{Z}$) and the other with $s+s'=2t+1,\ s-s'=2t'+1$.  Thus, we have
\begin{align}
&\eqref{eq:before_changing_var}\nonumber\\
\begin{split}
&= \frac{\bigl(2\pi^3\sigma_q^2\Lambda(\sigma_q^2,\sigma_p^2)\bigr)^{-\frac{1}{2}} \Gamma}{\sqrt{N_{\sigma_q^2,\sigma_p^2,\Gamma,j}N_{\sigma_q^2,\sigma_p^2,\Gamma,j'}}} \sum_{t,t'}\Biggl( \exp\left\{- \frac{1}{2\sigma_q^2} \left[q - \Gamma \left(t + \frac{j+j'}{2d} \right)\right]^2  - 2\sigma_q^2 p^2 - 2i\Gamma p \left( t'+\frac{j-j'}{2d}\right)  \right\}     \\
& \hspace{5.5cm} \times \exp\left\{-\frac{2\Gamma^2\sigma_p^2}{\Lambda(\sigma_q^2,\sigma_p^2)} \left[\left(t + \frac{j+j'}{2d}\right)^2 + \left(t' + \frac{j-j'}{2d}\right)^2\right] \right\} \\
& \hspace{4.8cm} + \exp\left\{ - \frac{1}{2\sigma_q^2} \left[q - \Gamma \left(t + \frac{j+j'}{2d} + \frac{1}{2} \right)\right]^2  -2\sigma_q^2 p^2 - 2i\Gamma p \left( t'+\frac{j-j'}{2d} + \frac{1}{2}\right)\right\}   \\ 
& \hspace{5.5cm} \times \exp\left\{-\frac{2\Gamma^2\sigma_p^2}{\Lambda(\sigma_q^2,\sigma_p^2)} \left[\left(t + \frac{j+j'}{2d} + \frac{1}{2}\right)^2 + \left(t' + \frac{j-j'}{2d} + \frac{1}{2}\right)^2\right] \right\} \Biggr) 
\end{split} \\
\begin{split}
&= \frac{\bigl(\pi^2 \Lambda(\sigma_q^2,\sigma_p^2)\bigr)^{-\frac{1}{2}} \Gamma}{\sqrt{N_{\sigma_q^2,\sigma_p^2,\Gamma,j}N_{\sigma_q^2,\sigma_p^2,\Gamma,j'}}} \Biggl\{ \biggl( E_{\frac{\Lambda(\sigma_q^2,\sigma_p^2)}{4\sigma_p^2},\Gamma,\frac{j+j'}{2d}}*G_{\sigma_q^2}(q) \biggr) e^{-2\sigma_q^2 p^2} \vartheta\! \left[ \begin{subarray}{c} \frac{j-j'}{2d} \\ \ \\ 0 \end{subarray}\right] \! \biggl(-\frac{\Gamma p}{\pi},\frac{2i\Gamma^2\sigma_p^2}{\pi \Lambda(\sigma_q^2,\sigma_p^2)} \biggr)  \\
& \hspace{4.5cm}  + \biggl( E_{\frac{\Lambda(\sigma_q^2,\sigma_p^2)}{4\sigma_p^2},\Gamma,\frac{j+j'}{2d} + \frac{1}{2}} * G_{\sigma_q^2}(q) \biggr)  e^{-2\sigma_q^2 p^2} \vartheta\! \left[ \begin{subarray}{c} \frac{j-j'}{2d} + \frac{1}{2} \\ \ \\ 0 \end{subarray}\right] \! \biggl(-\frac{\Gamma p}{\pi},\frac{2i\Gamma^2\sigma_p^2}{\pi \Lambda(\sigma_q^2,\sigma_p^2)} \biggr) \Biggr\} 
\end{split} \label{eq:wigner_theta_func}\\
\begin{split}
&= \frac{(2\pi \sigma_p^2)^{-\frac{1}{2}}}{\sqrt{N_{\sigma_q^2,\sigma_p^2,\Gamma,j}N_{\sigma_q^2,\sigma_p^2,\Gamma,j'}}} \Biggl\{ \biggl( E_{\frac{\Lambda(\sigma_q^2,\sigma_p^2)}{4\sigma_p^2},\Gamma,\frac{j+j'}{2d}}*G_{\sigma_q^2}(q) \biggr) e^{-\frac{p^2}{2\sigma_p^2}} \vartheta\! \left[ \begin{subarray}{c} 0 \\ \ \\ \frac{j-j'}{2d} \end{subarray}\right] \! \biggl( \frac{ip\Lambda(\sigma_q^2,\sigma_p^2)}{2\Gamma\sigma_p^2},\frac{\pi i  \Lambda(\sigma_q^2,\sigma_p^2)}{2\Gamma^2\sigma_p^2} \biggr) \\
& \hspace{4.5cm} + \biggl( E_{\frac{\Lambda(\sigma_q^2,\sigma_p^2)}{4\sigma_p^2},\Gamma,\frac{j+j'}{2d} + \frac{1}{2}}*G_{\sigma_q^2}(q) \biggr) e^{-\frac{p^2}{2\sigma_p^2}} \vartheta\! \left[ \begin{subarray}{c} 0 \\ \ \\ \frac{j-j'}{2d} + \frac{1}{2} \end{subarray}\right] \! \biggl( \frac{ip \Lambda(\sigma_q^2,\sigma_p^2)}{2\Gamma\sigma_q^2},\frac{\pi i  \Lambda(\sigma_q^2,\sigma_p^2)}{2\Gamma^2\sigma_p^2} \biggr)  \Biggr\} 
\end{split} \\
\begin{split}
&= \frac{(2\pi \sigma_p^2)^{-\frac{1}{2}}}{\sqrt{N_{\sigma_q^2,\sigma_p^2,\Gamma,j}N_{\sigma_q^2,\sigma_p^2,\Gamma,j'}}} \\
&\hspace{0.5cm} \times \sum_{t} \Biggl\{ \biggl( E_{\frac{\Lambda(\sigma_q^2,\sigma_p^2)}{4\sigma_p^2},\Gamma,\frac{j+j'}{2d}}*G_{\sigma_q^2}(q) \biggr) e^{2\pi i t \frac{j-j'}{2d} } \exp\Biggl[-\frac{1}{2\sigma_p^2}\biggl(p + \frac{\pi t \Lambda(\sigma_q^2,\sigma_p^2) }{\Gamma} \biggr)^2 -\frac{2\pi^2t^2 \sigma_q^2  \Lambda(\sigma_q^2,\sigma_p^2)}{\Gamma^2}\Biggr]  \\
& \hspace{1.3cm}  + \biggl( E_{\frac{\Lambda(\sigma_q^2,\sigma_p^2)}{4\sigma_p^2},\Gamma,\frac{j+j'}{2d}+\frac{1}{2}}*G_{\sigma_q^2}(q) \biggr) e^{2\pi i t \left(\frac{j-j'}{2d} + \frac{1}{2} \right)} \exp\Biggl[-\frac{1}{2\sigma_p^2}\biggl(p + \frac{\pi t \Lambda(\sigma_q^2,\sigma_p^2) }{\Gamma} \biggr)^2 -\frac{2\pi^2 t^2\sigma_q^2\Lambda(\sigma_q^2,\sigma_p^2)}{\Gamma^2 }\Biggr]\Biggr\}  \end{split} \\
\begin{split}
&= \frac{1}{\sqrt{ N_{\sigma_q^2,\sigma_p^2,\Gamma,j}N_{\sigma_q^2,\sigma_p^2,\Gamma,j'}}}\Biggl[ \biggl(E_{\frac{\Lambda(\sigma_q^2,\sigma_p^2)}{4\sigma_p^2},\Gamma,\frac{j+j'}{2d}}*G_{\sigma_q^2}(q) \biggr) \biggl(\tilde{E}_{\frac{\Lambda(\sigma_q^2,\sigma_p^2)}{4\sigma_q^2},\frac{\pi \Lambda(\sigma_q^2,\sigma_p^2)}{\Gamma},\frac{j-j'}{2d}} * G_{\sigma_p^2}(p)\biggr) \\ 
& \hspace{4.5cm} + \biggl( E_{\frac{\Lambda(\sigma_q^2,\sigma_p^2)}{4\sigma_p^2},\Gamma,\frac{j+j'}{2d} + \frac{1}{2}}*G_{\sigma_q^2}(q) \biggr) \biggl(\tilde{E}_{\frac{\Lambda(\sigma_q^2,\sigma_p^2)}{4\sigma_q^2},\frac{\pi \Lambda(\sigma_q^2,\sigma_p^2)}{\Gamma},\frac{j-j'}{2d} + \frac{1}{2}} * G_{\sigma_p^2}(p)\biggr) \Biggr],
\end{split}
\end{align}
where we used Eq.~\eqref{eq:theta_func_identity_2} in the third equality. \qed

\section{The proof of Proposition \ref{prop:ave_photon}} \label{sec:expect_quad}
In order to derive the average photon number of the approximate code state $\ket{j_{\sigma_q^2,\sigma_p^2,\Gamma}}$ in Definition \ref{def:standard_form}, we first calculate the expectation values $\braket{\hat{q}^2}_{\ket{j_{\sigma_q^2,\sigma_p^2,\Gamma}}}$ and $\braket{\hat{p}^2}_{\ket{j_{\sigma_q^2,\sigma_p^2,\Gamma}}}$ of the square of the quadrature operators $\hat{q}^2$ and $\hat{p}^2$ with respect to $\ket{j_{\sigma_q^2,\sigma_p^2,\Gamma}}$, using its Wigner function \eqref{eq:wigner_func}.  Then, one can obtain the average photon number $\braket{\hat{n}}_{\ket{j_{\sigma_q^2,\sigma_p^2,\Gamma}}}$ of the state $\ket{j_{\sigma_q^2,\sigma_p^2,\Gamma}}$ by exploiting the fact that $\braket{\hat{q}^2+\hat{p}^2}_{\ket{j_{\sigma_q^2,\sigma_p^2,\Gamma}}}=\braket{2\hat{n} + 1 }_{\ket{j_{\sigma_q^2,\sigma_p^2,\Gamma}}}  $. 
We frequently use Eqs.~\eqref{eq:integral_convolution}, \eqref{eq:integral_E}, \eqref{eq:integral_tilde_E}, and \eqref{eq:integral_normal} in the following calculation.
Let $\mathrm{Pr}_{\hat{q}}(q)$ and $\mathrm{Pr}_{\hat{p}}(p)$ be the probability densities to obtain the values $q$ and $p$ in the $\hat{q}$- and $\hat{p}$-quadrature measurements, respectively.  Then, they can be given by
\begin{align}
 \mathrm{Pr}_{\hat{q}}(q) &= \int dp\, W_{\ket{j_{\sigma_q^2,\sigma_p^2,\Gamma}}\bra{j_{\sigma_q^2,\sigma_p^2,\Gamma}}}(q,p)  = \frac{1}{N_{\sigma_q^2,\sigma_p^2,\Gamma,j}} \left[c_1 E_{\frac{\Lambda(\sigma_q^2,\sigma_p^2)}{4\sigma_p^2},\Gamma,\frac{j}{d}}  + c_2  E_{\frac{\Lambda(\sigma_q^2,\sigma_p^2)}{4\sigma_p^2},\Gamma,\frac{j}{d}+\frac{1}{2}} \right]  * G_{\sigma_q^2}(q), \label{eq:prob_dens_pos} \\
 \mathrm{Pr}_{\hat{p}}(p) &= \int dq\, W_{\ket{j_{\sigma_q^2,\sigma_p^2,\Gamma}}\bra{j_{\sigma_q^2,\sigma_p^2,\Gamma}}}(q,p) = \frac{1}{N_{\sigma_q^2,\sigma_p^2,\Gamma,j}} \left[c_3 \tilde{E}_{\frac{\Lambda(\sigma_q^2,\sigma_p^2)}{4\sigma_q^2},\frac{\pi \Lambda(\sigma_q^2,\sigma_p^2)}{\Gamma},0} + c_4\tilde{E}_{\frac{\Lambda(\sigma_q^2,\sigma_p^2)}{4\sigma_q^2},\frac{\pi \Lambda(\sigma_q^2,\sigma_p^2)}{\Gamma},\frac{1}{2}} \right] * G_{\sigma_p^2}(p) \label{eq:prob_dens_mom},
 \end{align}  
 where $c_1,c_2,c_3,$ and $c_4$ are defined as 
\begin{align}
 c_1&\coloneqq  \vartheta \! \left[\begin{subarray}{c} 0 \\ \ \\ 0 \end{subarray} \right]\!\left(0,2\pi i \Gamma^{-2}\sigma_q^2 \Lambda(\sigma_q^2,\sigma_p^2)\right),\\
 c_2&\coloneqq  \vartheta \! \left[\begin{subarray}{c} 0 \\ \ \\ \frac{1}{2} \end{subarray} \right]\!\left(0, 2\pi i \Gamma^{-2} \sigma_q^2 \Lambda(\sigma_q^2,\sigma_p^2)\right) ,\\
 c_3&\coloneqq  \vartheta \! \left[\begin{subarray}{c} \frac{j}{d} \\ \ \\ 0 \end{subarray} \right]\! \left(0,2\pi^{-1}i\Gamma^2\sigma_p^2 \bigl[\Lambda(\sigma_q^2,\sigma_p^2)\bigr]^{-1}\right), \\
 c_4&\coloneqq  \vartheta \! \left[\begin{subarray}{c} \frac{j}{d}+\frac{1}{2}\\ \ \\ 0 \end{subarray} \right]\! \left(0,2\pi^{-1}i\Gamma^2\sigma_p^2 \bigl[\Lambda(\sigma_q^2,\sigma_p^2)\bigr]^{-1}\right).
\end{align}
Note that the normalization constant $N_{\sigma_q^2,\sigma_p^2,\Gamma,j}$ satisfies $N_{\sigma_q^2,\sigma_p^2,\Gamma,j}=c_1c_3+c_2c_4$ as shown in Eq.~\eqref{eq:normalization}.
Using $\mathrm{Pr}_{\hat{q}}(q)$, we calculate the expectation value of $\hat{q}^2$ as follows:
\begin{align}
\braket{\hat{q}^2}_{\ket{j_{\sigma_q^2,\sigma_p^2,\Gamma}}} 
&= \int dq\, q^2\, \mathrm{Pr}_{\hat{q}}(q)  \\
&=\int dq \int dr\, \frac{q^2}{N_{\sigma_q^2,\sigma_p^2,\Gamma,j}} \left[c_1 E_{\frac{\Lambda(\sigma_q^2,\sigma_p^2)}{4\sigma_p^2},\Gamma,\frac{j}{d}}(r)  + c_2  E_{\frac{\Lambda(\sigma_q^2,\sigma_p^2)}{4\sigma_p^2},\Gamma,\frac{j}{d}+\frac{1}{2}}(r) \right] G_{\sigma_q^2}(q-r)  \\
&= \int dq \int dr\, \frac{r^2 + 2r(q-r) + (q-r)^2}{N_{\sigma_q^2,\sigma_p^2,\Gamma,j}}\left[c_1 E_{\frac{\Lambda(\sigma_q^2,\sigma_p^2)}{4\sigma_p^2},\Gamma,\frac{j}{d}}(r)  + c_2  E_{\frac{\Lambda(\sigma_q^2,\sigma_p^2)}{4\sigma_p^2},\Gamma,\frac{j}{d}+\frac{1}{2}}(r) \right]G_{\sigma_q^2}(q-r)\\
&= \int dq'\, q'^2 G_{\sigma_q^2}(q') + \int dr\, \frac{r^2}{N_{\sigma_q^2,\sigma_p^2,\Gamma,j}} \left[c_1 E_{\frac{\Lambda(\sigma_q^2,\sigma_p^2)}{4\sigma_p^2},\Gamma,\frac{j}{d}}(r)  + c_2  E_{\frac{\Lambda(\sigma_q^2,\sigma_p^2)}{4\sigma_p^2},\Gamma,\frac{j}{d}+\frac{1}{2}}(r) \right] \label{eq:separation}\\
\begin{split}
&= \sigma_q^2 + \int  \frac{dr}{N_{\sigma_q^2,\sigma_p^2,\Gamma,j}} \Biggl\{c_1 \left[\sum_{s\in\mathbb{Z}}\Gamma^2\left(s+\frac{j}{d}\right)^2 \exp\left[-\frac{\Gamma^2}{2\mu}\left(s+\frac{j}{d}\right)^2\right]\delta\left(r-\Gamma\left(s+\frac{j}{d}\right)\right) \right]  \\
& \hspace{1cm}  + c_2 \left[\sum_{s\in\mathbb{Z}}\Gamma^2\left(s+\frac{j}{d}+\frac{1}{2}\right)^2 \exp\left[-\frac{\Gamma^2}{2\mu}\left(s+\frac{j}{d}+\frac{1}{2}\right)^2\right]\delta\left(r-\Gamma\left(s+\frac{j}{d}+\frac{1}{2}\right)\right) \right] \Biggr\}\Biggr|_{\mu=\frac{\Lambda(\sigma_q^2,\sigma_p^2)}{4\sigma_p^2}} \end{split} \\
&= \sigma_q^2 - \frac{2}{N_{\sigma_q^2,\sigma_p^2,\Gamma,j}} \frac{\partial}{\partial( \mu^{-1})}\biggl[c_1  \vartheta \! \left[\begin{subarray}{c} \frac{j}{d} \\ \ \\ 0 \end{subarray} \right]\! \left(0, \frac{i\Gamma^2}{2\pi\mu} \right) + c_2 \vartheta \! \left[\begin{subarray}{c} \frac{j}{d}+\frac{1}{2} \\ \ \\ 0 \end{subarray} \right]\! \left(0, \frac{i\Gamma^2}{2\pi\mu}\right) \biggr]\biggr|_{\mu=\frac{\Lambda(\sigma_q^2,\sigma_p^2)}{4\sigma_p^2}} ,
\end{align}
where we used the fact that $G_{\sigma_q^2}(x)$ has zero mean in the fourth and the fifth equality.
In the same way, for the expectation value of $\hat{p}^2$, we have
\begin{equation}
\braket{\hat{p}^2}_{\ket{j_{\sigma_q^2,\sigma_p^2,\Gamma}}}
= \sigma_p^2 - \frac{2}{N_{\sigma_q^2,\sigma_p^2,\Gamma,j}} \frac{\partial}{\partial (\mu'^{-1})} \biggl[c_3 \vartheta \! \left[\begin{subarray}{c} 0 \\ \ \\ 0 \end{subarray} \right]\!\left(0,\frac{\pi i [\Lambda(\sigma_q^2,\sigma_p^2)\bigr]^2}{2\mu'\Gamma^2}\right) + c_4\vartheta \! \left[\begin{subarray}{c} 0 \\ \ \\ \frac{1}{2} \end{subarray} \right]\!\left(0, \frac{\pi i [\Lambda(\sigma_q^2,\sigma_p^2)\bigr]^2}{2\mu'\Gamma^2} \right) \biggr] \biggr|_{\mu'=\frac{\Lambda(\sigma_q^2,\sigma_p^2)}{4\sigma_q^2}}.
\end{equation}
Now we define $\tilde{N}_{\sigma_q^2,\sigma_p^2,\Gamma,j}(x,y)$ as
\begin{equation}
\tilde{N}_{\sigma_q^2,\sigma_p^2,\Gamma,j}(x,y) \coloneqq \vartheta \! \left[\begin{subarray}{c} \frac{j}{d} \\ \ \\ 0 \end{subarray} \right]\! \left(0,\frac{i\Gamma^2}{2\pi}x\right)\ \vartheta\!\left[\begin{subarray}{c} 0 \\ \ \\ 0 \end{subarray} \right]\! \left(0,\frac{\pi i \bigl[\Lambda(\sigma_q^2,\sigma_p^2)\bigr]^2}{2\Gamma^2} y \right) + \vartheta \! \left[\begin{subarray}{c} \frac{j}{d}+\frac{1}{2} \\ \ \\ 0 \end{subarray} \right]\! \left(0,\frac{i\Gamma^2}{2\pi}x\right)\ \vartheta\!\left[\begin{subarray}{c} 0 \\ \ \\ \frac{1}{2} \end{subarray} \right]\!\left(0, \frac{\pi i \bigl[\Lambda(\sigma_q^2,\sigma_p^2)\bigr]^2}{2\Gamma^2} y \right),
\end{equation}  
where $N_{\sigma_q^2,\sigma_p^2,\Gamma,j} = \tilde{N}_{\sigma_q^2,\sigma_p^2,\Gamma,j}\left(4\sigma_p^2\bigl[\Lambda(\sigma_q^2,\sigma_p^2)\bigr]^{-1},4\sigma_q^2\bigl[\Lambda(\sigma_q^2,\sigma_p^2)\bigr]^{-1}\right)$.
Then, the average photon number is given by
\begin{align}
\braket{n}_{\ket{j_{\sigma_q^2,\sigma_p^2,\Gamma}}}&= \frac{\braket{\hat{q}^2 + \hat{p}^2}_{\ket{j_{\sigma_q^2,\sigma_p^2,\Gamma}}} -1 }{2} \nonumber\\
&= \frac{\sigma_q^2+\sigma_p^2-1}{2} - \left(\frac{\partial}{\partial x} + \frac{\partial}{\partial y}\right) \ln \tilde{N}_{\sigma_q^2,\sigma_p^2,\Gamma,j}(x,y) \Biggr|_{x=4\sigma_p^2 \bigl[\Lambda(\sigma_q^2,\sigma_p^2)\bigr]^{-1},y=4\sigma_q^2 \bigl[\Lambda(\sigma_q^2,\sigma_p^2)\bigr]^{-1}},
\end{align} 
which proves Eq.~\eqref{eq:ave_photon_num}. \qed

\section{Alternative expressions for the Wigner function, inner product, and average photon number} \label{sec:alternative_expression}
In this appendix, we derive alternative expressions for the Wigner function, inner products, and the average photon number of the standard form $\ket{j_{\sigma_q^2,\sigma_p^2,\Gamma}}$ in terms of multivariable generalization of the theta function, the Riemann theta function (also called Siegel theta function) \cite{Mumford2007}.

For $\vec{z}\in\mathbb{C}^n$ and $\bm{\tau}\in\mathbb{C}^n\times \mathbb{C}^n$ with $\bm{\tau}=\bm{\tau}^{\top}$ and $\mathrm{Im}(\bm{\tau})>0$, the Riemann theta function $\Theta \! \left[\begin{subarray}{c} \vec{a} \\ \ \\ \vec{b} \end{subarray} \right]\! (\vec{z},\bm{\tau})$, is defined as
\begin{equation}
    \Theta \! \left[\begin{subarray}{c} \vec{a} \\ \ \\ \vec{b} \end{subarray} \right]\! (\vec{z},\bm{\tau})  \coloneqq \sum_{\vec{s}\in\mathbb{Z}^n}\exp\bigl[\pi i (\vec{s}+\vec{a})^{\top}\bm{\tau} (\vec{s}+\vec{a}) + 2\pi i (\vec{z}+\vec{b})^{\top}\! \cdot (\vec{s}+\vec{a})\bigr],
\end{equation}
where $\cdot$ denotes an inner product.
We also define multivariate normal distribution $\check{G}[\bm{\nu}](\vec{x})$ as
\begin{equation}
    \check{G}[\bm{\nu}](\vec{x})\coloneqq \frac{1}{\sqrt{2\pi \mathrm{det}(\bm{\nu})}}\exp\Bigl(-\frac{1}{2}\vec{x}^{\top}\bm{\nu}^{-1}\vec{x}\Bigr).
\end{equation}
Now we define a multivariable function combining $E_{\mu,\Gamma,a}(x)$ and $\tilde{E}_{\mu,\Gamma,a}(x)$ as follows.
\begin{definition}\label{def:check_E}
    For a symmetric $2\times 2$ matrix $\bm{\mu}$ satisfying $\mathrm{Re}(\bm{\mu})>0$ and 2-dimensional vectors $\vec{\Gamma},\vec{a},$ and $\vec{b}$, let $\check{E}\bigl[\bm{\mu},\vec{\Gamma},\vec{a},\vec{b}\bigr](\vec{x})$ be defined as
    \begin{equation}
        \check{E}\bigl[\bm{\mu},\vec{\Gamma},\vec{a},\vec{b}\bigr](\vec{x}) \coloneqq \exp\Bigl(-\frac{1}{2}\vec{x}^{\top}\bm{\mu}^{-1}\vec{x}\Bigr) \sum_{\vec{s} \in\mathbb{Z}^2} e^{2\pi i \vec{b}^{\top}\!\cdot (\vec{s}+\vec{a})}\; \delta\Bigl(\vec{x}-(\vec{s}+\vec{a})\circ\vec{\Gamma}\Bigr),
    \end{equation}
    where $\circ$ denotes an Hadamard product $(A\circ B)_{ij}=(A)_{ij}(B)_{ij}$.
\end{definition}

Under Definition \ref{def:check_E}, we have an alternative expression of Wigner function \eqref{eq:wigner_func}.
\begin{cor}[Alternative expression of Wigner function]\label{cor:wigner_alt}
The Wigner function given in Eq.~\eqref{eq:wigner_func} is alternatively represented as
\begin{equation}
    \begin{split}
    &W_{\ket{j_{\sigma_q^2,\sigma_p^2,\Gamma}}\bra{j'_{\sigma_q^2,\sigma_p^2,\Gamma}}}(q,p) \\ 
    &= \frac{1}{\sqrt{ N_{\sigma_q^2,\sigma_p^2,\Gamma,j}N_{\sigma_q^2,\sigma_p^2,\Gamma,j'}}} \check{E}\biggl[ \Lambda(\sigma_q^2,\sigma_p^2)\left(\begin{smallmatrix} {4\sigma_p^2} & 2 i \\ 2 i & 4\sigma_q^2  \end{smallmatrix}\right)^{-1}\! , \bigl(\tfrac{\Gamma}{2},\tfrac{\pi \Lambda(\sigma_q^2,\sigma_p^2)}{\Gamma}\bigr)^{\top}, \bigl(\tfrac{j+j'}{d}, 0\bigr)^{\top}, \bigl(0, \tfrac{j'}{d})^{\top}\biggr]\, \vec{*} \; \check{G}\Bigl[ \left(\begin{smallmatrix} \sigma_q^2 & 0 \\ 0 & \sigma_p^2 \end{smallmatrix}\right) \Bigr]\bigl((q,p)^{\top}\bigr), 
    \end{split}
    \label{eq:wigner_multi}
\end{equation}
where $\check{E}$ is defined in Definition \ref{def:check_E}, and $\vec{*}$ denotes a convolution in a multivariate sense.
\end{cor}

\begin{proof}
    Comparing Eqs.~\eqref{eq:wigner_func} and \eqref{eq:wigner_multi}, it is sufficient to show the following equality.
    \begin{equation}
        \begin{split}
        &\biggl(E_{\mu,\Gamma,a}*G_{\sigma_q^2}(q) \biggr) \biggl(\tilde{E}_{\mu',\Gamma',a'} * G_{\sigma_p^2}(p)\biggr) + \biggl( E_{\mu,\Gamma,a+\frac{1}{2}}*G_{\sigma_q^2}(q) \biggr) \biggl(\tilde{E}_{\mu',\Gamma',a'+\frac{1}{2}} * G_{\sigma_p^2}(p)\biggr) \\
        &= \check{E}\biggl[\left(\begin{smallmatrix} \mu^{-1} & \frac{2\pi i}{\Gamma\Gamma'} \\ \frac{2\pi i}{\Gamma\Gamma'} & \mu'^{-1} \end{smallmatrix}\right)^{-1}\! , \bigl(\tfrac{\Gamma}{2},\Gamma'\bigr)^{\top}, \bigl(2a, 0\bigr)^{\top}, \bigl(0, a-a')^{\top}\biggr]\, \vec{*} \; \check{G}\Bigl[ \left(\begin{smallmatrix} \sigma_q^2 & 0 \\ 0 & \sigma_p^2 \end{smallmatrix}\right) \Bigr]\bigl((q,p)^{\top}\bigr).
        \end{split}
    \end{equation}
    This follows from the following rearrangement of the summation.
    \begin{align}
        &\biggl(E_{\mu,\Gamma,a}*G_{\sigma_q^2}(q) \biggr) \biggl(\tilde{E}_{\mu',\Gamma',a'} * G_{\sigma_p^2}(p)\biggr) + \biggl( E_{\mu,\Gamma,a+\frac{1}{2}}*G_{\sigma_q^2}(q) \biggr) \biggl(\tilde{E}_{\mu',\Gamma',a'+\frac{1}{2}} * G_{\sigma_p^2}(p)\biggr) \nonumber \\
        \begin{split}
        &= \iint dx dy\, \sum_{s,s'}\Biggl[\exp\Bigl(-\frac{x^2}{2\mu}\Bigr) \delta\bigl(x-(s+a)\Gamma\bigr) G_{\sigma_q^2}(q-x) \exp\Bigl(-\frac{y^2}{2\mu'}+2\pi i a' s'\Bigr)\delta(y+ s'\Gamma' ) G_{\sigma_p^2}(p-y) \\
        &\hspace{2cm} + \exp\Bigl(-\frac{x^2}{2\mu}\Bigr)\delta\Bigl(x-\bigl(s+a+\frac{1}{2}\bigr)\Gamma\Bigr) G_{\sigma_q^2}(q-x) \exp\Bigl(-\frac{y^2}{2\mu'}+2\pi i \bigl(a'+\frac{1}{2}\bigr) s' \Bigr)\delta(y+s' \Gamma' ) G_{\sigma_p^2}(p-y) \Biggr]
        \end{split} \\
        \begin{split}
        &= \iint dx dy\, \sum_{s,s'}\Biggl[\exp\Bigl(-\frac{x^2}{2\mu}\Bigr) \delta\Bigl(x-(2s+2a)\frac{\Gamma}{2}\Bigr) G_{\sigma_q^2}(q-x) \exp\Bigl(-\frac{y^2}{2\mu'}-2\pi i a' s' - \pi i (2s) s'\Bigr)\delta(y-s'\Gamma' ) G_{\sigma_p^2}(p-y) \\
        &\hspace{1.3cm} + \exp\Bigl(-\frac{x^2}{2\mu}\Bigr)\delta\Bigl(x-\bigl(2s+1+2a\bigr) \frac{\Gamma}{2}\Bigr) G_{\sigma_q^2}(q-x) \exp\Bigl(-\frac{y^2}{2\mu'}-2\pi i a's'-\pi i (2s+1) s' \Bigr)\delta(y- s'\Gamma') G_{\sigma_p^2}(p-y) \Biggr]
        \end{split} \\
        &= \iint dx dy\, \sum_{s',s''} \exp\Bigl[-\frac{x^2}{2\mu}-\frac{y^2}{2\mu'}-\pi i s'' s' -2\pi i a' s' \Bigr]\delta\Bigl(x-(s''+2a)\frac{\Gamma}{2}\Bigr)\delta(y-s'\Gamma')\; \check{G}\Bigl[ \left(\begin{smallmatrix} \sigma_q^2 & 0 \\ 0 & \sigma_p^2 \end{smallmatrix}\right) \Bigr]\bigl((q-x, p-y)^{\top}\bigr)\\
        \begin{split}
        &= \iint dx dy\, \sum_{s',s''} \exp\Bigl[-\frac{1}{2}(x, y)\left(\begin{smallmatrix} \mu^{-1} & \frac{2\pi i}{\Gamma \Gamma'} \\ \frac{2\pi i}{\Gamma \Gamma'} & \mu'^{-1} \end{smallmatrix}\right)(x, y)^{\top} + 2\pi i (a-a')s'\Bigr]\delta\Bigl(x-(s''+2a)\frac{\Gamma}{2}\Bigr)\delta(y-s'\Gamma') \\
        &\hspace{7cm} \times \check{G}\Bigl[ \left(\begin{smallmatrix} \sigma_q^2 & 0 \\ 0 & \sigma_p^2 \end{smallmatrix}\right) \Bigr]\bigl((q-x, p-y)^{\top}\bigr)
        \end{split} \\
        &= \check{E}\biggl[\left(\begin{smallmatrix} \mu^{-1} & \frac{2\pi i}{\Gamma\Gamma'} \\ \frac{2\pi i}{\Gamma\Gamma'} & \mu'^{-1} \end{smallmatrix}\right)^{-1}\! , \bigl(\tfrac{\Gamma}{2},\Gamma'\bigr)^{\top}, \bigl(2a, 0\bigr)^{\top}, \bigl(0, a-a')^{\top}\biggr]\, \vec{*} \; \check{G}\Bigl[ \left(\begin{smallmatrix} \sigma_q^2 & 0 \\ 0 & \sigma_p^2 \end{smallmatrix}\right) \Bigr]\bigl((q,p)^{\top}\bigr).
    \end{align}
\end{proof}
Note that the right-hand side of Eq.~\eqref{eq:wigner_multi} approaches the right-hand side of Eq.~\eqref{eq:wigner_ideal_multi} as $\sigma_q^2,\sigma_p^2\rightarrow 0$.  
Equation~\eqref{eq:wigner_multi} fits a viewpoint that a state corresponding to the Wigner function $\check{E}(q,p)$ is subject to Gaussian random displacement channel \cite{Caruso2006}, since random displacement can be represented as a convolution in the Wigner function picture.  This viewpoint is utilized in numerical simulations of error analyses using approximate GKP codes \cite{Menicucci2014,Fukui2018,Vuillot2019,Noh2019,Wang2019}.  It should be noted that an operator corresponding to $\check{E}(q,p)$ with parameters chosen as in Eq.~\eqref{eq:wigner_multi} is neither a density operator nor a limit of density operators.  There is thus no contradiction with the observation that an approximate GKP state differs from an ideal code state subject to random displacement noise, as stated in the explanation below the definition of Approximation 2.

\begin{cor}[Alternative expressions of normalization constant and inner product]\label{cor:inner_alt} 
    The normalization factor given in Eq.~\eqref{eq:normalization} is alternatively represented as
    \begin{equation}
        N_{\sigma_q^2,\sigma_p^2,\Gamma,j} = \Theta \! \left[ \begin{subarray}{c} \bigl(\tfrac{2j}{d}, 0\bigr)^{\top} \\ \ \\ \bigl(0, \tfrac{j}{d}\bigr)^{\top} \end{subarray}\right] \! \Biggl( \vec{0}, \left(\begin{smallmatrix}\frac{i\sigma_p^2\Gamma^2}{2\pi \Lambda(\sigma_q^2,\sigma_p^2)} & -\frac{1}{2} \\ -\frac{1}{2} & \frac{2\pi i\sigma_q^2\Lambda(\sigma_q^2,\sigma_p^2)}{\Gamma^2} \end{smallmatrix}\right) \Biggr).
    \end{equation}
    Furthermore, the inner product given in Eq.~\eqref{eq:inner_product} is alternatively represented as
    \begin{equation}
        \braket{j'_{\sigma_q^2,\sigma_p^2,\Gamma}|j_{\sigma_q^2,\sigma_p^2,\Gamma}} = \frac{1}{\sqrt{N_{\sigma_q^2,\sigma_p^2,\Gamma,j} N_{\sigma_q^2,\sigma_p^2,\Gamma,j'}}} \Theta \! \left[ \begin{subarray}{c} \bigl(\tfrac{j+j'}{d}, 0\bigr)^{\top} \\ \ \\ \bigl(0, \tfrac{j'}{d}\bigr)^{\top} \end{subarray}\right] \! \Biggl( \vec{0}, \left(\begin{smallmatrix}\frac{i\sigma_p^2\Gamma^2}{2\pi \Lambda(\sigma_q^2,\sigma_p^2)} & -\frac{1}{2} \\ -\frac{1}{2} & \frac{2\pi i\sigma_q^2\Lambda(\sigma_q^2,\sigma_p^2)}{\Gamma^2} \end{smallmatrix}\right) \Biggr).
    \end{equation}
\end{cor}
\begin{proof}
    We combine Eq.~\eqref{eq:inner_trace} with the followings:
    \begin{align}
        \int d\vec{x}\, f(\vec{x})\, \vec{*}\, g(\vec{x}) &= \int d\vec{x}\, f(\vec{x}) \int d\vec{y}\, g(\vec{y}), \\
        \begin{split}
        \int d\vec{x}\, \check{E}[\bm{\mu},\vec{\Gamma},\vec{a},\vec{b}](\vec{x}) &= \sum_{\vec{s}\in\mathbb{Z}^2} \exp\Bigl[-\tfrac{1}{2}\bigl((\vec{s}+\vec{a})\circ \vec{\Gamma}\bigr)^{\top}\bm{\mu}^{-1}\bigl((\vec{s}+\vec{a})\circ \vec{\Gamma}\bigr) + 2\pi i \vec{b}^{\top}\! \cdot \! (\vec{s}+\vec{a})\Bigr] \\
        &= \Theta \! \left[ \begin{subarray}{c} \vec{a} \\ \ \\ \vec{b} \end{subarray}\right] \! \biggl( \vec{0}, \frac{i}{2\pi}\bm{\mu}^{-1}\circ \left(\begin{smallmatrix}\Gamma_{1}^2 & \Gamma_1\Gamma_2 \\ \Gamma_1\Gamma_2 & \Gamma_2^2 \end{smallmatrix}\right) \biggr), \end{split} \label{eq:integrate_check_E}\\
        \int d\vec{x}\, \check{G}[\bm{\nu}](\vec{x})&=1.
    \end{align}
\end{proof}

\begin{cor}[Alternative expression of average photon number]\label{cor:ave_photon_num_alt}
    The average photon number given in Eq.~\eqref{eq:ave_photon_num} is alternatively represented as
    \begin{equation}
        \braket{\hat{n}}_{\ket{j_{\sigma_q^2,\sigma_p^2,\Gamma}}} 
        = \frac{\sigma_q^2+\sigma_p^2 - 1}{2} -\left(\frac{\partial}{\partial (\bm{\mu}^{-1})_{11}} + \frac{\partial}{\partial (\bm{\mu}^{-1})_{22}}\right) \ln \check{N}_{\sigma_q^2,\sigma_p^2,\Gamma,j}(\bm{\mu}^{-1})\Biggr|_{\bm{\mu}^{-1}=\frac{1}{\Lambda(\sigma_q^2,\sigma_p^2)}\left(\begin{smallmatrix} {4\sigma_p^2} & 2 i \\ 2 i & 4\sigma_q^2  \end{smallmatrix}\right)},
    \end{equation}
    where $\check{N}_{\sigma_q^2,\sigma_p^2,\Gamma,j}(\bm{\mu}^{-1})$ is given by
    \begin{equation}
        \check{N}_{\sigma_q^2,\sigma_p^2,\Gamma,j}(\bm{\mu}^{-1})\coloneqq \Theta \! \left[ \begin{subarray}{c} \bigl(\frac{2j}{d}, 0\bigr)^{\top} \\ \ \\ \bigl(0, \frac{j}{d}\bigr)^{\top} \end{subarray}\right] \! \Biggl( \vec{0}, \frac{i}{2\pi}\bm{\mu}^{-1}\circ \left(\begin{smallmatrix}\frac{\Gamma^2}{4} & \frac{\pi\Lambda(\sigma_q^2,\sigma_p^2)}{2} \\ \frac{\pi\Lambda(\sigma_q^2,\sigma_p^2)}{2} & \frac{\pi^2[\Lambda(\sigma_q^2,\sigma_p^2)]^2}{\Gamma^2} \end{smallmatrix}\right) \Biggr).
    \end{equation}
\end{cor}
\begin{proof}
    We give a more intuitive proof rather than direct calculation given in Sec.~\ref{sec:expect_quad}.  We see that $\braket{\hat{q}^2 + \hat{p}^2}_{\ket{j_{\sigma_q^2,\sigma_p^2,\Gamma}}}$ denotes a second moment of the (quasi)probability distribution $W_{\ket{j_{\sigma_q^2,\sigma_p^2,\Gamma}}\bra{j_{\sigma_q^2,\sigma_p^2,\Gamma}}}$, which is a convolution of normalization times $\check{E}$ and $\check{G}[\bm{\nu}]$.  Since $\check{G}[\bm{\nu}](\vec{x})$ has zero mean, the second moment of $W_{\ket{j_{\sigma_q^2,\sigma_p^2,\Gamma}}\bra{j_{\sigma_q^2,\sigma_p^2,\Gamma}}}$ is a summation of the second moment of normalization times $\check{E}$ and $\check{G}[\bm{\nu}]$.  (Eq.~\eqref{eq:separation} also shows this fact.)   The second moment of $\check{G}[\bm{\nu}]$ is simply given by $\sigma_q^2+\sigma_p^2$.  On the other hand, the second moment of normalization times $\check{E}$ is given by
    \begin{align}
        &\frac{1}{N_{\sigma_q^2,\sigma_p^2,\Gamma,j}}\int d\vec{x}\, \|\vec{x}\|^2 \check{E}\biggl[ \Lambda(\sigma_q^2,\sigma_p^2)\left(\begin{smallmatrix} {4\sigma_p^2} & 2 i \\ 2 i & 4\sigma_q^2  \end{smallmatrix}\right)^{-1}\! , \bigl(\tfrac{\Gamma}{2},\tfrac{\pi \Lambda(\sigma_q^2,\sigma_p^2)}{\Gamma}\bigr)^{\top}, \bigl(\tfrac{2j}{d}, 0\bigr)^{\top}, \bigl(0, \tfrac{j}{d})^{\top}\biggr](\vec{x}) \\
        &= -\frac{2}{N_{\sigma_q^2,\sigma_p^2,\Gamma,j}}\int d\vec{x}\, \left(\frac{\partial}{\partial (\bm{\mu}^{-1})_{11}} + \frac{\partial}{\partial (\bm{\mu}^{-1})_{22}}\right) \check{E}\biggl[ \bm{\mu} , \bigl(\tfrac{\Gamma}{2},\tfrac{\pi \Lambda(\sigma_q^2,\sigma_p^2)}{\Gamma}\bigr)^{\top}, \bigl(\tfrac{2j}{d}, 0\bigr)^{\top}, \bigl(0, \tfrac{j}{d})^{\top}\biggr](\vec{x})\Biggr|_{\bm{\mu}=\Lambda(\sigma_q^2,\sigma_p^2)\left(\begin{smallmatrix} {4\sigma_p^2} & 2 i \\ 2 i & 4\sigma_q^2  \end{smallmatrix}\right)^{-1}} \\
        &= -2 \left(\frac{\partial}{\partial (\bm{\mu}^{-1})_{11}} + \frac{\partial}{\partial (\bm{\mu}^{-1})_{22}}\right) \ln \Theta \! \left[ \begin{subarray}{c} \bigl(\frac{2j}{d}, 0\bigr)^{\top} \\ \ \\ \bigl(0, \frac{j}{d}\bigr)^{\top} \end{subarray}\right] \! \Biggl( \vec{0}, \frac{i}{2\pi}\bm{\mu}^{-1}\circ \left(\begin{smallmatrix}\frac{\Gamma^2}{4} & \frac{\pi\Lambda(\sigma_q^2,\sigma_p^2)}{2} \\ \frac{\pi\Lambda(\sigma_q^2,\sigma_p^2)}{2} & \frac{\pi^2[\Lambda(\sigma_q^2,\sigma_p^2)]^2}{\Gamma^2} \end{smallmatrix}\right) \Biggr)\Biggr|_{\bm{\mu}^{-1}=\frac{1}{\Lambda(\sigma_q^2,\sigma_p^2)}\left(\begin{smallmatrix} {4\sigma_p^2} & 2 i \\ 2 i & 4\sigma_q^2  \end{smallmatrix}\right)},
    \end{align}
    where we used Eq.~\eqref{eq:integrate_check_E} in the last equality.  Combining these with the relation $\braket{\hat{q}^2+\hat{p}^2}_{\ket{j_{\sigma_q^2,\sigma_p^2,\Gamma}}}=\braket{2\hat{n} + 1 }_{\ket{j_{\sigma_q^2,\sigma_p^2,\Gamma}}}  $ proves the statement.
\end{proof}

\end{widetext}

\bibliography{equivalence_approx_gkp}

\end{document}